\documentclass[runningheads,USenglish]{llncs}

\usepackage{silence}
\ActivateWarningFilters[pdftoc]

\usepackage[utf8]{inputenc}
\usepackage[T1]{fontenc}
\usepackage{complexity}
\usepackage{enumerate}
\usepackage{graphicx}
\graphicspath{{figures/}}
\usepackage{amsfonts}
\usepackage{amssymb}
\usepackage{amsmath}
\usepackage{wasysym}
\usepackage{mathtools}
\usepackage{bibentry}
\usepackage{comment}
\usepackage{paralist}
\usepackage[inline]{enumitem}
\usepackage{todonotes}
\usepackage{xspace}
\usepackage{url}
\usepackage{esvect}

\usepackage{xcolor}
\definecolor{defblue}{rgb}{0.121,0.47,0.705}
\definecolor{linkblue}{rgb}{0.098,0.098,0.4392}
\let\emph\relax
\DeclareTextFontCommand{\emph}{\color{defblue}\em}
\DeclareTextFontCommand{\bl}{\color{defblue}}
\usepackage[colorlinks=true,
    linkcolor=linkblue,
    anchorcolor=linkblue,
    citecolor=linkblue,
    filecolor=linkblue,
    menucolor=linkblue,
    urlcolor=linkblue,
    bookmarks=true,
    bookmarksopen=true,
    bookmarksopenlevel=2,
    bookmarksnumbered=true,
    hyperindex=true,
    plainpages=false,
    pdfpagelabels=true
]{hyperref}
\usepackage[capitalise,noabbrev,nameinlink]{cleveref}
\usepackage{subcaption}

\usepackage{float}
\usepackage{thmtools,apptools,thm-restate}
\newcommand{\restateref}[1]{\IfAppendix{\hyperref[#1]{$\star$}}{\hyperref[#1*]{$\star$}}}

\let\doendproof\endproof
\renewcommand\endproof{~\hfill$\qed$\doendproof}

\makeatletter
\let\orgdescriptionlabel\descriptionlabel
\renewcommand*{\descriptionlabel}[1]{%
  \let\orglabel\label
  \let\label\@gobble
  \phantomsection
  \edef\@currentlabel{#1}%
  \let\label\orglabel
  \orgdescriptionlabel{#1}%
}
\makeatother

\renewcommand{\orcidID}[1]{\href{https://orcid.org/#1}{\includegraphics[scale=.03]{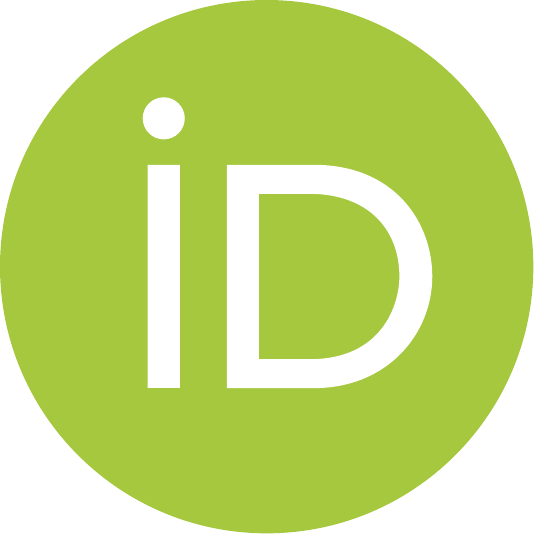}}}

\makeatletter
\spn@wtheorem{observation}{Observation}{\bfseries}{\itshape}
\spn@wtheorem{numclaim}{Claim}{\bfseries}{\itshape}
\makeatother
\crefname{numclaim}{Claim}{Claims}
\Crefname{numclaim}{Claim}{Claims}
\crefname{observation}{Observation}{Observations}
\Crefname{observation}{Observation}{Observations}
\crefname{condition}{Condition}{Conditions}
\Crefname{condition}{Condition}{Conditions}

\newcommand{\mclique}{\textsc{Multicolored Clique}\xspace} 
\newcommand{\AoNF}{\textsc{All or Nothing Flow}\xspace}
\newcommand{\AoNFs}{\textsc{AoNF}\xspace} 
\newcommand{\COr}{\textsc{Circulating Orientation}\xspace}

\newcommand{\Oh}{\mathcal{O}}
\DeclareMathOperator{\pw}{pw}
\DeclareMathOperator{\VS}{VS}

\def\calG{\mathcal G}
\newcommand{\calF}{\mathcal{F}}
\DeclareMathOperator{\tw}{tw}
\let\CH\relax
\DeclareMathOperator{\CH}{CH}

\newcommand{\Epsilon}{\ensuremath{\mathcal{E}}}

\DeclareMathOperator{\rep}{rep}

\newcommand{\defquestion}[3]{
	\vspace{2mm}
	\noindent\fbox{
		\begin{minipage}{0.96\linewidth}
			\begin{tabular*}{\linewidth}{@{\extracolsep{\fill}}lr} \textsc{#1} & \\ \end{tabular*}
			{\bf{Input:}} #2 \\
			{\bf{Question:}} #3
		\end{minipage}
	}
	\vspace{2mm}
}

\newcommand{\defparquestion}[4]{
	\vspace{2mm}
	\noindent\fbox{
		\begin{minipage}{0.96\linewidth}
			\begin{tabular*}{\linewidth}{@{\extracolsep{\fill}}lr} \textsc{#1} & \\ \end{tabular*}
			{\bf{Input:}} #2 \\
                {\bf{Parameter:}} #3 \\
			{\bf{Question:}} #4
		\end{minipage}
	}
	\vspace{2mm}
}

\newcounter{casecounter}
\newcounter{subcasecounter}

\newcounter{subsubcasecounter}


\crefname{casecounter}{Case}{Cases}
\crefname{subcasecounter}{Case}{Cases}
\crefname{subsubcasecounter}{Case}{Cases}
\Crefname{casecounter}{Case}{Cases}
\Crefname{subcasecounter}{Case}{Cases}
\Crefname{subsubcasecounter}{Case}{Cases}

\begin{document}
\crefname{observation}{Observation}{Observations}
\Crefname{observation}{Observation}{Observations}
\crefname{numclaim}{Claim}{Claims}
\Crefname{numclaim}{Claim}{Claims}

\title{Upward and Orthogonal Planarity are W[1]-hard Parameterized by Treewidth}
\titlerunning{Upward and Orthogonal Planarity are W[1]-hard by Treewidth}

\author{Bart M. P. Jansen\inst{1}\orcidID{0000-0001-8204-1268} 
\and Liana Khazaliya\inst{2} \orcidID{0009-0002-3012-7240} 
\and Philipp Kindermann\inst{3} \orcidID{0000-0001-5764-7719} 
\and \\Giuseppe Liotta\inst{4} \orcidID{0000-0002-2886-9694} 
\and Fabrizio Montecchiani\inst{4}\orcidID{0000-0002-0543-8912} 
\and Kirill Simonov\inst{5}\orcidID{0000-0001-9436-7310}
}
\authorrunning{B.M.P. Jansen et al.}

\institute{Eindhoven University of Technology, Eindhoven, The Netherlands\footnote{Bart M. P. Jansen has received funding from the European Research Council (ERC) under the European Union’s Horizon 2020 research and innovation programme (grant agreement No 803421, ReduceSearch).} \email{b.m.p.jansen@tue.nl} 
\and Technische Universit\"{a}t Wien, Vienna, Austria\footnote{Liana Khazaliya is supported by Vienna Science and Technology Fund (WWTF) [10.47379/ICT22029]; Austrian Science Fund (FWF) [Y1329]; European Union’s Horizon 2020 COFUND programme [LogiCS@TUWien,
grant agreement No. 101034440].} \email{lkhazaliya@ac.tuwien.ac.at}
\and Universit\"{a}t Trier, Trier, Germany \email{kindermann@uni-trier.de} 
\and University of Perugia, Perugia, Italy\footnote{This work was supported, in part, by MUR of Italy, under PRIN Project n. 2022ME9Z78 - NextGRAAL: Next-generation algorithms for constrained GRAph visuALization, and under PRIN Project n. 2022TS4Y3N - EXPAND: scalable algorithms for EXPloratory Analyses of heterogeneous and dynamic Networked Data.}\email{name.surname@unipg.it}
\and Hasso Plattner Institute, University of Potsdam, Potsdam, Germany\footnote{Kirill Simonov acknowledges support by DFG Research Group ADYN via grant DFG 411362735.} \email{kirillsimonov@gmail.com}
}

\maketitle

\begin{abstract}
\textsc{Upward planarity testing} and \textsc{Rectilinear planarity testing} are central problems in graph drawing. It is known that they are both \NP-complete, but \XP~when parameterized by treewidth. In this paper we show that these two problems are W[1]-hard parameterized by treewidth, which answers open problems posed in two earlier papers. The key step in our proof is an analysis of the \textsc{All-or-Nothing Flow} problem, a generalization of which was used as an intermediate step in the NP-completeness proof for both planarity testing problems. We prove that the flow problem is W[1]-hard parameterized by treewidth on planar graphs, and that the existing chain of reductions to the planarity testing problems can be adapted without blowing up the treewidth. Our reductions also show that the known $n^{\Oh(\mathsf{tw})}$-time algorithms  cannot be improved to run in time~$n^{o(\mathsf{tw})}$ unless \textsc{ETH} fails.
\end{abstract}

%\vskip1cm

\section{Introduction}
A graph is \emph{planar} if it admits a drawing in the plane where no two edges cross each other. Testing graph planarity is among the most fundamental problems in graph algorithms and  graph drawing. 
%The research on this subject can be classified into two main directions.
%On one side, several papers proposed efficient algorithms for this problem (the most famous being the linear-time testing algorithm of Hopcroft and Tarjan~\cite{DBLP:journals/jacm/HopcroftT74}), on the other side, several variants and restrictions have been proposed, including clustered planarity (see, e.g.~\cite{DBLP:conf/esa/BlasiusFR21,DBLP:journals/jacm/FulekT22}), constrained planarity (see, e.g.~\cite{DBLP:journals/talg/BlasiusR16,DBLP:journals/jcss/LiottaRT23}), and $k$-planarity (see, e.g.~\cite{DBLP:journals/algorithmica/GrigorievB07,KM-2012,DBLP:journals/ipl/UrschelW21}); also refer to~\cite{DBLP:reference/crc/Patrignani13} for a survey.
While several papers proposed efficient algorithms for this problem (including the celebrated linear-time algorithm of Hopcroft and Tarjan~\cite{DBLP:journals/jacm/HopcroftT74}), notable variants and restrictions have also been investigated,  including clustered planarity (see, e.g.~\cite{DBLP:conf/esa/BlasiusFR21,DBLP:journals/jacm/FulekT22}), constrained planarity (see, e.g.~\cite{DBLP:journals/talg/BlasiusR16,DBLP:journals/jcss/LiottaRT23}), and $k$-planarity (see, e.g.~\cite{DBLP:journals/algorithmica/GrigorievB07,KM-2012,DBLP:journals/ipl/UrschelW21});  refer to~\cite{DBLP:reference/crc/Patrignani13} for a survey.

\begin{figure}
    \centering
    \subcaptionbox{Directed graph $\vv{G}$\label{fig:intro-input}}[.32\textwidth]{\includegraphics[page=1]{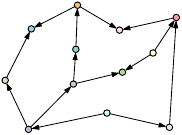}}
    \hfil
    \subcaptionbox{Upward planar drawing\label{fig:intro-upward}}[.32\textwidth]{\includegraphics[page=2]{example}}
    \hfil
    \subcaptionbox{Rectilinear drawing\label{fig:intro-rectilinear}}[.32\textwidth]{\includegraphics[page=3]{example}}
    \caption{A directed graph $\protect\vv{G}$, an upward planar drawing of $\protect\vv{G}$, and a rectilinear planar drawing of its underlying undirected graph $G$.}
    \label{fig:intro}
\end{figure}

This paper investigates two of such classical variants, namely \emph{upward planarity testing} and \emph{rectilinear planarity testing}. Given a directed acyclic graph $G$, upward planarity testing asks whether $G$ admits a crossing-free drawing where all edges are monotonically increasing in a common direction, which is conventionally called the upward direction; see \cref{fig:intro-upward}. For an undirected graph $G$, rectilinear planarity testing asks whether $G$ admits a crossing-free drawing such that each edge is either a horizontal or a vertical segment; see \cref{fig:intro-rectilinear}. Both upward planarity and rectilinear planarity testing are classical and extensively investigated topics in graph drawing (see, e.g.~\cite{DBLP:books/ph/BattistaETT99,DBLP:conf/dagstuhl/1999dg,DBLP:books/ws/NishizekiR04,DBLP:reference/crc/2013gd}).

While apparently different, the two problems have a lot in common. Namely, both an upward planar drawing of a digraph and a rectilinear planar drawing of a graph exists provided that the graph has a planar embedding where for every face $f$ there is some ``balancing'' of the angles that the edges along the boundary of the face form in the interior of $f$. Consider, for simplicity, biconnected graphs. In a rectilinear planar drawing of a biconnected graph the boundary of every internal face $f$ is an orthogonal polygon and hence the number of $\frac{\pi}{2}$ angles  minus the number of $\frac{3\pi}{2}$ angles must be $4$. Similarly, if a biconnected graph is upward planar it admits an upward planar drawing where for every internal face the number of angles that are smaller than $\frac{\pi}{2}$ always exceeds by two units the number of angles that are larger than $\pi$.

Hence, it is quite natural to see that in the \emph{fixed embedding setting} (i.e. when the combinatorial structure of the faces of the graph is given as part of the input) classical results solve both upward planarity testing and rectilinear planarity testing in polynomial time by looking for the existence of a feasible flow in a network where vertices supply angles to faces and faces have a limited capacity which depends on structure of the graph~\cite{DBLP:journals/algorithmica/BertolazziBLM94,DBLP:journals/siamcomp/Tamassia87}.
On the other hand, both problems are \NP-complete in the so-called \emph{variable embedding setting}, that is when the testing algorithm must verify whether the input graph has a combinatorial structure of its faces which allows the  balancing of angles described above. Again unsurprisingly, both proofs of \NP-completeness follow the same logic based on a reduction from a common flow problem on planar graphs~\cite{GargT01}.

These \NP-completeness results have motivated a flourishing literature describing both polynomial-time solutions for special classes of graphs and parameterized solutions for general graphs. For example, polynomial-time solutions are known for both problems when the input graph has treewidth at most two~\cite{DBLP:journals/siamcomp/BattistaLV98,DBLP:journals/siamdm/DidimoGL09,DBLP:conf/gd/DidimoKLO22,DBLP:journals/comgeo/Frati22,DBLP:journals/jcss/GiacomoLM22}; also, rectilinear planarity testing can be solved in linear time if the maximum degree of the input graph is at most three~\cite{DBLP:conf/soda/DidimoLOP20}, and upward planarity testing can be solved in linear time if the digraph has only one source vertex~\cite{DBLP:journals/siamcomp/BertolazziBMT98,DBLP:conf/gd/BrucknerHR19,DBLP:journals/siamcomp/HuttonL96}.
Concerning parameterized solutions, upward planarity testing is fixed-parameter tractable when parameterized by the number of triconnected components~\cite{10.1007/978-3-540-30140-0_16}, by the treedepth~\cite{DBLP:conf/compgeom/ChaplickGFGRS22}, and by the number of sources~\cite{DBLP:conf/compgeom/ChaplickGFGRS22}.

The research in this paper is motivated by the fact that both upward planarity and rectilinear planarity testing are known to lie in   \XP~when parameterized by  treewidth~\cite{DBLP:conf/compgeom/ChaplickGFGRS22,DBLP:journals/jcss/GiacomoLM22}. Determining whether these two parameterized problems are in \FPT~are mentioned as open problems in both~\cite{DBLP:conf/compgeom/ChaplickGFGRS22,DBLP:journals/jcss/GiacomoLM22}. The main contribution of this paper is as follows.

\begin{theorem}\label{th:main}
Upward planarity testing and rectilinear planarity testing parameterized by treewidth are both \textsc{W[1]}-hard. Moreover, assuming the Exponential Time Hypothesis, neither problem can be solved in time~$f(k) \cdot n^{o(k)}$ for any computable function~$f$, where~$k$ is the treewidth of the input graph.
\end{theorem}

\cref{th:main} implies that, under the standard hypothesis \FPT~$\neq$~\W[1] in parameterized complexity, there exists no fixed-parameter tractable algorithm for either problem parameterized by treewidth, hence answering the above mentioned open problems.
To obtain our results we analyze the auxiliary flow problem used as a common starting point in the \NP-completeness proof of both planarity  problems. It closely resembles the \textsc{All-or-Nothing Flow} problem (\AoNFs), which asks for an $st$-flow of prescribed value in an edge-capacitated flow network such that each edge is either used fully, or not at all. The \AoNFs problem parameterized by treewidth was recently shown to be \textsc{W[1]}-hard (in fact, even \textsc{XNLP}-complete) on general graphs by Bodlaender et al.~\cite{BodlaenderCW22}. By a significant adaptation of their construction, we can prove that \AoNFs parameterized by treewidth remains \textsc{W[1]}-hard on \emph{planar} graphs. By revisiting the chain of reductions to the planarity testing problems, passing through the \textsc{Circulating Orientation} problem in between, we show they can be carried out without blowing up the treewidth of the graph and thereby obtain~\cref{th:main}.

The rest of the paper is organized as follows. In \cref{sec:preliminaries} we formally define the problems involved in our chain of reductions. Due to space limitations, all formal proofs have been deferred to the appendix. In \cref{sec:overview} we therefore provide high-level sketches of our proofs. We conclude in \cref{sec:conclusion}. For space reasons, results marked with a (clickable) ``$\star$'' are proved in the appendix.

\section{Preliminaries} \label{sec:preliminaries}
We assume familiarity with the basic notions of graph drawing~\cite{DBLP:books/ph/BattistaETT99} and of parameterized complexity~\cite{CyganFKLMPPS15}, including the notions of treewidth and pathwidth which are commonly used parameters to capture the complexity of a graph~$G$. The formal definitions, including those for upwards and rectilinear planarity, can be found in the appendix. Below we define the parameterized problems which are used in the chain of reductions of our lower bounds. Throughout the paper we utilize both undirected and directed graphs, which may have parallel edges but no loops. A graph without parallel arcs or edges is a \emph{simple} graph. We use $uv$ to denote a (directed) arc from~$u$ to~$v$ and~$\{u,v\}$ to denote an (undirected) edge between~$u$ and~$v$. The vertex set of a graph~$G$ is denoted by~$V(G)$ and the (multi)set of edges by~$E(G)$. Depending on whether the graph is directed or not,~$E(G)$ either contains its undirected edges or its directed arcs. For a vertex in a directed graph~$G$, we denote by~$E^-_G(v)$ the (multi)set of arcs leading into~$v$, and by~$E^+_G(v)$ the (multi)set of arcs leading out of~$v$.

The starting point of our reductions is the following parameterized version of the \textsc{Clique} problem, which is well-known to be W[1]-complete~\cite[Thm. 13.25]{CyganFKLMPPS15}. Assuming ETH, it cannot be solved in time~$f(k) n^{o(k)}$ for any computable function~$f$~\cite[Cor. 14.32]{CyganFKLMPPS15}.

\defparquestion{\mclique}{An undirected simple graph~$G$ and a partition of its vertex set into~$k$ sets~$V_1, \ldots, V_k$, each consisting of~$N$ vertices.}{$k$.}{Does~$G$ contain a clique~$C \subseteq V(G)$ such that~$|C \cap V_i| = 1$ for each~$i \in [k]$?}

%Note that we assume that the sets~$V_i$ in the input of the problem are independent sets in~$G$. This is without loss of generality, since a solution is allowed to contain only one vertex for each such set. 
Note that the assumption that all sets~$V_i$ have the same size is without loss of generality, since we may pad the input with isolated vertices if needed.

The next problem in our chain of reductions is a variation of \textsc{Maximum Flow}; we therefore need some terminology regarding flows. A \emph{flow network}~$(G,c,s,t)$ consists of a directed graph~$G$ with a capacity function~$c \colon E(G) \to \mathbb{Z}_+$ on the arcs, together with two distinct vertices~$s,t$ called the \emph{source} and \emph{sink}. We allow a flow network to have parallel arcs. An \emph{$st$-flow} in the flow network is a function~$f \colon E(G) \to \mathbb{Z}_{\geq 0}$ such that for each arc~$e \in E(G)$ we have~$0 \leq f(e) \leq c(e)$ (capacity constraints), and for each vertex~$v \in V(G) \setminus \{s,t\}$, we have~$\sum_{e \in E^-_G(v)} f(e) = \sum_{e \in E^+_G(v)} f(e)$ (flow conservation). The \emph{value} of the flow is defined as~$\sum_{e \in E^+_G(s)} f(e) - \sum _{e \in E^-_G(s)} f(e)$. For a vertex~$v$ in a flow network~$(G,c,s,t)$, we denote the total capacity of arcs leaving~$v$ by~$d^+_G(v) := \sum_{e \in E^+_G(v)} c(e)$ and use~$d^-_G(v) := \sum_{e \in E^+_G(v)} c(e)$ for the total capacity of arcs entering~$v$. While a maximum flow can be found in polynomial time, the following variation is hard.

\defquestion
{\AoNF}
{A flow network~$(G,c,s,t)$ and a positive integer~$\calF$.}
{Does there exist an $st$-flow of value exactly~$\calF$, such that the flow through any arc~$uv \in E(G)$ is either~$0$ or equal to~$c(uv)$?}

We then reduce to the following problem on undirected graphs that models the combinatorial difficulty encountered in testing upwards or rectilinear planarity, because it captures the problem of deciding orientations of edges to balance certain contributions around a vertex (i.e. a face of the dual graph). 

\defquestion
{\COr}
{An undirected graph~$G$ with an edge-capacity function~$c \colon E(G) \to \mathbb{Z}_{\geq 0}$.}
{Is it possible to orient the edges of~$G$, such that for each vertex~$v \in V(G)$ the total capacity of edges oriented into~$v$ is equal to the total capacity of edges oriented out of~$v$? (Such an orientation is called a \emph{circulating orientation}.)} 

Edges are allowed to have capacity~$0$ in this problem, which allows us to construct \emph{triconnected} instances in the hardness reduction by inserting capacity-0 edges that do not violate planarity and do not blow up the pathwidth. For an undirected multigraph~$G$, we use~$E_G(v)$ to denote the (multi)set of edges incident on a vertex~$v \in V(G)$. In the context of an edge-capacity function~$c$, we denote the total capacity of edges incident on~$v$ by~$d_G(v) := \sum_{e \in E_G(v)} c(e)$. %\lk{check weight/capacity consistency for CO}

In the non-planar setting, \AoNF easily reduces to \COr by a polynomial-time transformation that increases the pathwidth by only a constant: it suffices to add a super-source and super-sink with properly chosen capacities on their incident edges~\cite{BodlaenderCW22}. Since the addition of a super-source and super-sink typically violates planarity of the graph, in our hardness construction for the flow problem we take special care to produce instances that can later be reduced to \COr without violating planarity. We point out that Didimo et al.~\cite{DidimoLP19} recently proved the NP-completeness of \COr on planar graphs. However, their reduction does not have any consequences for the complexity of the problem parameterized by treewidth.

% \defquestion
% {Upward Planarity Testing}
% {A directed planar graph~$G$.}
% {Does~$G$ admit a crossing-free planar drawing where all edges are monotonically increasing in a common direction?} 

% \defquestion
% {Rectilinear Planarity Testing}
% {An undirected planar graph~$G$ of maximum degree at most four.}
% {Does~$G$ admit a crossing-free planar drawing where each edge is a horizontal or vertical line segment?} 

\section{Overview} \label{sec:overview}

Since space requirements prohibit us from presenting our reductions in detail, we give an outline that discusses the main technical ideas behind our result and defer the formal proofs to the appendix. % We first focus on showing hardness of \textsc{Circulating Orientation} on planar instances, and then on transferring the hardness from \textsc{Circulating Orientation} to \textsc{Upward Planarity Testing} and \textsc{Rectilinear Planarity Testing}.

\subsection{Hardness of \textsc{All-or-Nothing-Flow}}
\paragraph{The non-planar case.} To aid the intuition, we first sketch an FPT-reduction from \mclique to non-planar \AoNF on a graph of pathwidth~$\Oh(k)$, which is inspired by an XNLP-completeness proof due to Bodlaender et al.~\cite{BodlaenderCW22}. Given an input~$(G, V_1, \ldots, V_k, k)$, which asks whether~$G$ has a clique containing exactly one vertex from each of the size-$N$ sets~$V_1, \ldots, V_k$, we construct a flow network~$\calG$ as follows. Number the vertices in each set~$V_i$ as~$v_{i,1}, \ldots, v_{i,N}$. Let~$m = |\overline{E}(G)|$, where~$\overline{E}(G)$ is the set of unordered vertex pairs which do not form an edge of~$G$. The graph~$\calG$ contains~$k$ \emph{rows}~$R_1, \ldots, R_k$. Each row~$R_i$ consists of~$m+1$ vertices~$V_{i}^{j}$ for~$j\in [m+1]$. For each~$j \in [m]$, there are~$N$ parallel arcs from~$V_{i}^{j}$ to~$V_{i}^{j+1}$ whose capacities are~$2kN+2q$ for each~$q \in [N]$. (Below, we will decrease the capacities of some of these arcs by~$1$, to model the non-edges of~$G$.) Intuitively, sending flow over an arc with capacity~$2kN+2q$ on row~$R_i$ corresponds to selecting the~$q$-th vertex of~$V_i$ into the clique. Flow conservation will ensure that the same choice is made for all arcs on the same row~$R_i$. To feed each row~$R_i$, there is an arc of capacity~$2kN$ from the source~$s$ to the first vertex~$V_{i}^{1}$, along with~$N$ parallel arcs of capacity~$2$. This is sufficient to saturate any single edge on the row, but insufficient to saturate two edges. The analogous arcs leave from the last vertex~$V_{i}^{m+1}$ to the sink~$t$.

\begin{figure}[t]
    \centering
    \includegraphics[width=\textwidth]{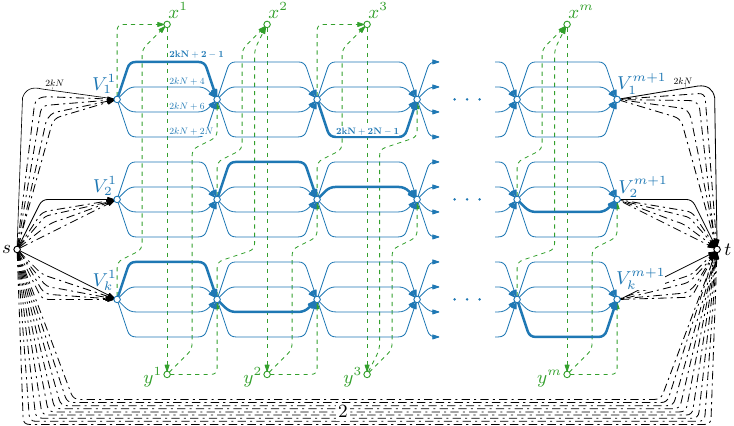}
    \caption{Illustration for the FPT-reduction from \mclique with~$k=3$,$N=4$ to non-planar \AoNF on a graph of pathwidth~$\Oh(k)$. Dashed edges have capacity~1, dash-dotted edges ($k(N-1)$ of them on the bottom) have capacity~2.}
    \label{fig:multicolored-k-clique-reduction}
\end{figure}

We will set the target value~$\calF$ of the flow problem to~$k(2kN+2N)$, which is effectively the amount that we obtain when selecting the largest index on each row. To compensate for the fact that some rows may want to select a smaller index, corresponding to sending less than~$2kN+2N$ flow, we additionally add $k(N-1)$ parallel capacity-2 arcs directly from~$s$ to~$t$. This number is chosen so that any feasible flow which sends at least value~$2kN+2$ through each row can be augmented to a flow of value exactly~$k(2kN+2N)$, while making it impossible to reach the desired flow value without choosing one index on each row.

The final part of the construction ensures that the choices encoded by the flow correspond to a multicolored $k$-clique in~$G$, using a small gadget that crucially exploits the all-or-nothing property of the flow; see \cref{fig:multicolored-k-clique-reduction}. For each pair~$v_{i,a}, v_{\ell,b}$ of non-adjacent vertices of~$G$ with~$i \neq \ell$, we incorporate a gadget to ensure that we cannot simultaneously choose flow value~$2kN+2a$ on row~$R_i$ and~$2kN+2b$ on row~$R_{\ell}$. We pick a unique index~$z$ from~$[m]$ for this non-edge of~$G$, and adapt the construction as follows:
\begin{itemize}
    % \item We consider the copies of the arcs $V_{i}^{t}V_{i}^{t+1}$ and~$V_{\ell}^{t} V_{\ell}^{t+1}$ which correspond to choosing vertices~$v_{i,a}$ and~$v_{\ell,b}$, whose capacities were set to~$2kN + 2a$ and~$2kN + 2b$. We decrease the capacity of both arcs by~$1$.
    \item Among the multiple arcs from~$V_i^z$ to~$V_i^{z+1}$, we consider the arc whose capacity was set to~$2kN + 2a$, which corresponds to choosing vertex~$v_{i,a}$. We decrease the capacity of this arc by~$1$.
    \item Among the multiple arcs from~$V_\ell^z$ to~$V_\ell^{z+1}$, we consider the arc whose capacity was set to~$2kN + 2b$, which correspond to choosing vertex~$v_{\ell,b}$. We decrease the capacity of this arc by~$1$.
    \item We introduce two new vertices $x^z, y^z$ and an arc~$x^z y^z$ of capacity~$1$. Vertices~$V_{i}^{z}, V_{\ell}^{z}$ both get an arc to~$x^z$ of capacity~$1$. Vertices~$V_{i}^{z+1}, V_{\ell}^{z+1}$ both get an arc from~$y^z$ of capacity~$1$.
\end{itemize}
The idea behind the gadget is as follows. It is possible to send a flow of value~$2kN+2a$ from~$V_{i}^{z}$ to~$V_{i}^{z+1}$ (by utilizing the arc of capacity~$2kN+2a-1$ to~$V_{i}^{z+1}$ along with a flow of value 1 along the path via~$x^z, y^z$). Analogously, it is possible to send a flow of value~$2kN+2b$ from~$V_{\ell}^{z}$ to~$V_{\ell}^{z+1}$. But we cannot do both simultaneously, due to the capacity-1 bottleneck between~$x^z$ and~$y^z$ and the fact that the capacity of the other arcs on the row is either too large or too small to be used in an all-or-nothing fashion. Hence we ensure the values~$a$ and~$b$ cannot be simultaneously selected on rows~$i$ and~$\ell$.

The construction is completed by inserting such a gadget for each non-edge of~$G$. The resulting flow network~$\calG$ can be shown to have pathwidth~$\Oh(k)$ since it effectively consists of~$k$ paths whose interconnections are confined to vertices whose index differ by at most one along the path.

\paragraph{Planarizing the instance.} It is conceptually not difficult to extend the construction to prove hardness also for planar instances of the flow problem, since the all-or-nothing nature of the flow facilitates a simple method to eliminate edge crossings. Suppose we have a crossing between two arcs~$uv$ and~$xy$ whose capacities are different. Then we may simply replace the crossing arcs~$uv,xy$ by a new vertex~$d$ of degree four along with arcs~$ud, dv$ of capacity~$c(uv)$ and arcs~$xd, dy$ of capacity~$c(xy)$. This transformation preserves the answer to the \AoNF problem: in one direction, any flow in the original network trivially yields a flow of the same value in the transformed network. The interesting step is the converse direction. An all-or-nothing flow in the reduced network either sends flow over both halves~$ud,dv$ (respectively~$xd,dy$) of an arc, or over neither half of the arc: flow-conservation ensures that all flow entering~$d$ must also exit~$d$, while the all-or-nothing property of the flow together with the fact that~$c(uv) \neq c(xy)$ means that flow entering on arc~$ud$ ($xd$) cannot leave on arc~$dy$ ($dv$). (The same approach for eliminating edge crossings was used previously by Didimo et al.~\cite{DidimoLP19}.)

Since the non-planar drawings coming out of the construction above (\cref{fig:multicolored-k-clique-reduction}) have the property that all crossings involve pairs of arcs of different capacities, we can simply planarize the drawing by inserting degree-four vertices where needed. It is not difficult to show that the pathwidth increases by only a constant factor, resulting in the following lemma. Its proof can be found in the appendix.

\begin{restatable}[\restateref{lemma:aonf}]{lemma}{LemAonf}
\label{lemma:aonf}
There is a polynomial-time algorithm that, given an instance of \mclique with parameter~$k$, outputs an equivalent instance of \AoNF on a planar graph of pathwidth~$\Oh(k)$ whose edge capacities are bounded by a polynomial in~$|V(G)|$.
\end{restatable}

The bound on the edge capacities of the instance will later govern the size of \emph{tendril} gadgets that will be created for the planarity testing problems.

\subsection{Hardness of \textsc{Circulating Orientation}}
We continue by describing the relation between \AoNF and \COr, starting with a special case that will be insightful to establish some intuition. Suppose we have a flow network~$(G,s,t,c)$ in which we ask for an all-or-nothing flow of value~$\calF$, satisfying the following conditions: for each vertex~$v \in V(G) \setminus \{s,t\}$ we have~$d^+_G(v) = d^-_G(v)$, the source has no incoming arcs, the sink has no outgoing arcs, and~$\calF = d^+_G(s)/2 = d^-_G(t)/2$. 

We argue that this flow instance is equivalent to the instance of \COr on the edge-capacitated undirected graph~$G'$ that is simply obtained from~$G$ by dropping the orientation of the edges. An all-or-nothing flow~$f$ of value~$\calF$ in~$G$ leads to a circulating orientation of the undirected graph~$G'$, as follows: simply start from the orientation of arcs as given by~$G$, but reverse the orientation for each arc~$uv \in E(G)$ with~$f(uv) = 0$. For each vertex~$v \in V(G) \setminus \{s,t\}$, flow conservation ensures that the total capacity of the incoming edges that carry flow is equal to the capacity of outgoing edges that carry flow. By the assumption that~$d^-_G(v) = d^+_G(v)$, the total capacity of out-arcs of~$v$ which is reversed by the process  equals that of the capacity of in-arcs of~$v$ which is reversed. Hence the total capacity of arcs which are oriented outwards remains unaffected by the reversals and equals~$d^-_G(v) = d^+_G(v) = d_{G'}(v)/2$. For the source~$s$, since the flow has value~$\calF = d^+(s)/2$ by assumption, reversing the orientation of edges not carrying flow leaves~$\calF = d^+_G(s)/2$ of capacity oriented out of~$s$, which is exactly half of~$d_{G'}(s)$ since there are no incoming arcs. The situation for the sink is analogous. This shows that a solution to the flow problem yields a solution to \COr, and it is not difficult to show the converse also holds.

Under the simplifying assumptions above, it is therefore trivial to reduce from \AoNF to \COr. In the flow network~$G$ we construct above, most vertices~$v \in V(G) \setminus \{s,t\}$ satisfy~$d^-_G(v) = d^+_G(v)$. The only exceptions are the~$x^j$-vertices, along with the last vertex of each row (their in-capacity exceeds their out-capacity), and the~$y^j$ vertices along with the first vertex of each row (their out-capacity exceeds their in-capacity). In general graphs, the imbalance can be resolved by adding a super-source~$S$ and super-sink~$T$. Here we can do something similar while preserving planarity, utilizing the fact that there is a face in the embedding that contains all~$x^j$ vertices along with the source~$s$ and sink~$t$,  and another face that contains all~$y^j$-vertices and~$\{s,t\}$. We insert a super-source and super-sink into these faces and use them to resolve the imbalance of the~$x^j$ and~$y^j$ vertices. Based on a delicate argument, we show that the imbalance of the first and last vertices of each row can be resolved via the standard source and sink, which are already adjacent to them.

For the super-source~$S$ and super-sink~$T$ to work as desired, we need an edge between them. This cannot be drawn in a planar fashion. However, we argue that the effect of the edge~$\{S,T\}$ can be simulated by having a four-cycle~$(s,S,t,T)$ involving the standard source and sink with appropriate capacities, which can be added in a planar fashion. By carefully setting the capacities, this four-cycle also resolves the issue caused by the fact that~$\calF \neq d^+_G(s)/2$.

To transform the resulting instances of \COr into the two planarity testing problems, it will be useful for the constructed planar graph to be triangulated (which implies it is triconnected). We can achieve this property by a post-processing step based on a result by Biedl. She proved~\cite{Biedl16} that any simple planar graph~$G$ of pathwidth~$k$ can be transformed into a triangulated planar supergraph~$G'$ on the same vertex set having pathwidth~$\Oh(k)$, and such a triangulation can be computed efficiently. To make our graph simple, we can start by subdiving all edges which is known to increase the pathwidth by at most a constant. In the context of the \COr problem, in which we are allowed to have edges of capacity 0 which do not affect the answer to the problem, we may then compute a triangulation of the simple graph and assign all newly introduced edges capacity 0, thereby leading to the following lemma.

\begin{restatable}[\restateref{lemma:co}]{lemma}{LemCo}
\label{lemma:co}
There is a polynomial-time algorithm that, given an instance of \mclique with parameter~$k$, outputs an equivalent instance of \COr on a simple, triconnected, triangulated planar graph of pathwidth~$\Oh(k)$ whose edge capacities are bounded by a polynomial in~$|V(G)|$.
\end{restatable}

\subsection{From \textsc{Circulating Orientation} to planarity testing problems}

The main idea behind the reduction from \COr to upward/rectilinear planarity testing is largely the same as in the original \NP-hardness proof of Garg and Tamassia~\cite{GargT01}. In what follows we explain the approach and highlight the differences.
Since the reductions for both problems are fairly similar, we mostly focus on \textsc{Upward Planarity Testing}. See also~\cref{fig:enter-label} for illustration.

\begin{figure}
    \centering
    \subcaptionbox{$P$\label{fig:upward-P}}{\includegraphics[page=1,width=0.35\textwidth]{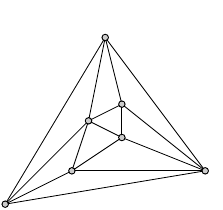}}\hfil
    \subcaptionbox{$D$\label{fig:upward-D}}{\includegraphics[page=2,width=0.35\textwidth]{./figures/upward.pdf}}\vspace{2mm}
    \subcaptionbox{$\vv{D}$\label{fig:upward-VD}}{\includegraphics[page=3,width=0.35\textwidth]{./figures/upward.pdf}}\hfil
    \subcaptionbox{$T_3$\label{fig:upward-T}}{\includegraphics[page=4,width=0.35\textwidth]{./figures/upward.pdf}}\vspace{2mm}
    \subcaptionbox{$\vv{G}$\label{fig:upward-G}}{\includegraphics[page=5,width=0.65\textwidth]{./figures/upward.pdf}}
    \caption{The reduction to \textsc{Upward Planarity Testing}: (a) A triangulated planar graph $P$. 
    (b) The dual $D$ of $P$ (based on the depicted planar embedding). 
    (c) An orientation  $\protect\vv{D}$ of $D$ which is an $st$-planar graph. 
    (d) A tendril $T_3$ and its schematization; the red boundary has negative contribution, while the blue (dashed) boundary has positive contribution. 
    (e) Construction of $\protect\vv{G}$: replacing the edges of a face of $\protect\vv{D}$ with the corresponding tendrils.}
    \label{fig:enter-label}
\end{figure}

As per the previous subsection, we start our reduction with an instance $(P, c)$ of \textsc{Circulating Orientation}, where $P$ is a planar triconnected graph of pathwidth $\Oh(k)$. The first step is to consider the dual graph $D$ of $P$. By known results~\cite{AminiHP09} about the pathwidth of planar graphs, the pathwidth of $D$ is also $\Oh(k)$. We also consider the graph $D$ to be weighted by $c$: the weight $w(e)$ of an edge $e$ in $D$ is set to be the capacity of its dual edge in $P$.
%Then we set an orientation $\vv{D}$ of $D$ so that the resulting graph is $st$-planar, and together with triconnectivity this gives a unique upward planar embedding of $\vv{D}$.
We obtain the final digraph $\vv{G}$ of the reduction as follows: every edge $e \in E(D)$ is replaced by a \emph{tendril} $T_{w(e)}$. The tendrils $T_\ell$ are special gadget graphs designed by Garg and Tamassia~\cite{GargT01}; their properties are: (i) the upward planar embedding is unique; and (ii) one of the boundary walks has contribution $2\ell$ to the adjacent face, and the other boundary walk has contribution $-2\ell$. Here, contribution refers to the angle assignment characterization of upward planarity: roughly speaking, the graph admits an upward planar embedding if and only if the angles in a planar embedding of the graph could be assigned numbers in $\{-1, 0, 1\}$ according to certain rules so that the sum of angles on the boundary of every inner face is $-2$, and on the boundary of the outer face is $2$.

    \begin{figure}[h]
        \centering
        \includegraphics[page=6, width=0.5\textwidth]{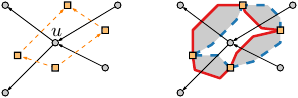}
        \caption{Constructing a circulating orientation of $(P,c)$.}
        \label{fig:backtop}
    \end{figure}

Now, since the ``skeleton'' graph $D$ is triconnected and the inserted tendrils are also triconnected, the planar embeddings of $G$ are essentially defined by the flip of each tendril. Picking the flip of a tendril then directly corresponds to picking the orientation of the respective edge in $P$. Specifically, the property of the target orientation of $P$ is that the sum of the weights of outgoing and incoming edges of a vertex is zero; this translates to the property that the edges of $D$ are to be oriented so that for every face, the sum of the weights of the clockwise edges is equal to the sum of the weights of the counter-clockwise edges. Finally, the latter translates to the upward planarity condition on the face of $\vv{G}$ that originates from a face of $\vv{D}$: the weights of clockwise and counter-clockwise edges are balanced if and only if the total contribution of all tendrils to the face is zero, since the contribution of a tendril is proportional to the weight of the respective edge, and the sign is picked exactly in accordance with the orientation of the edge. See~\cref{fig:backtop} that illustrates the correspondence.

As opposed to the original \NP-hardness proof~\cite{GargT01}, our starting point is the more special \textsc{Circulating Orientation} problem, so we only require the fixed-contribution tendril gadgets, and not the wiggle gadgets that have variable contribution. We also start the reduction with an arbitrary planar triconnected graph, instead of the special instance originating from the satisfiability problem; all in all this leads to an arguably clearer and more direct \NP-hardness proof. We additionally observe that the reduction performed in this way does not blow up the pathwidth, which is crucial for our main result.

Finally, there are a few differences in the case of \textsc{Rectilinear Planarity Testing}. First, it is important that we start with a triangulated graph $P$ and so the graph $D$ is of maximum degree 3. The edges of the graph $D$ are then subdivided to obtain the graph $F$ that admits a rectilinear embedding. Finally, to obtain the target graph $G$, one edge in every subdivision is replaced by a \emph{rectilinear tendril}~\cite{GargT01}, which play an analogous role to that of tendrils above. In the same way, faces of $G$ originating from faces of $D$ correspond to vertices in $P$, with the contribution of tendrils to the face being proportional to the capacities of the respective edges.

Combining these transformations with \cref{lemma:co}, we obtain the following.

\begin{restatable}[\restateref{lemma:planarity}]{lemma}{LemPlanarity}
\label{lemma:planarity}
There is a polynomial-time algorithm that, given an instance of \mclique with parameter~$k$, outputs an equivalent instance of \textsc{Upward planarity testing} (respectively \textsc{Rectilinear planarity testing}) on a graph of pathwidth~$\Oh(k)$.
\end{restatable}

Due to known ETH-based lower bounds for the W[1]-complete \mclique problem and the well-known fact that the treewidth of a graph is not larger than its pathwidth, \cref{th:main} follows directly from \cref{lemma:planarity}.

\section{Conclusion} \label{sec:conclusion}
We proved that \textsc{Upward planarity} and \textsc{Rectilinear planarity} are both W[1]-hard parameterized by treewidth and that the $n^{\Oh(\mathsf{tw})}$ running times of the existing  algorithms for them are tight assuming ETH. Our reduction also provides an alternative NP-completeness proof for these problems, which avoids the use of the so-called \emph{wiggle} gadgets~\cite{GargT01}.

The \textsc{All-or-Nothing Flow} problem on general graphs was recently shown to be XNLP-complete~\cite{BodlaenderCW22} parameterized by treewidth. This complexity class (which contains $W[1]$) was recently introduced~\cite{BodlaenderGNS21} and captures parameterized problems solvable in nondeterministic FPT-time and logarithmic space. It would be interesting to see whether the two planarity testing problems parameterized by treewidth are also XNLP-complete.

Our results show that the parameter treewidth is too general to allow for FPT algorithms for the considered problems. An investigation of more restrictive parameterizations that yield fixed-parameter tractability is left for future work. Based on preliminary investigations, we believe both \textsc{Upward planarity} and \textsc{Rectilinear planarity} may be FPT parameterized by the cutwidth of the dual multigraph. We remark that, since the instances produced by our hardness reduction have pathwidth~$\Oh(k)$ and maximum degree~$\Oh(1)$, the cutwidth of the primal graph is also~$\Oh(k)$ (\cite[Thm. 49]{Bodlaender98}). Hence our hardness results extend to the parameterization by the cutwidth of the primal graph.

\subsubsection{Acknowledgements}
We acknowledge the fruitful working atmosphere of Dagstuhl Seminar 23162 ``New Frontiers of Parameterized Complexity in Graph Drawing'', where this work was started.

\includegraphics[height=2cm]{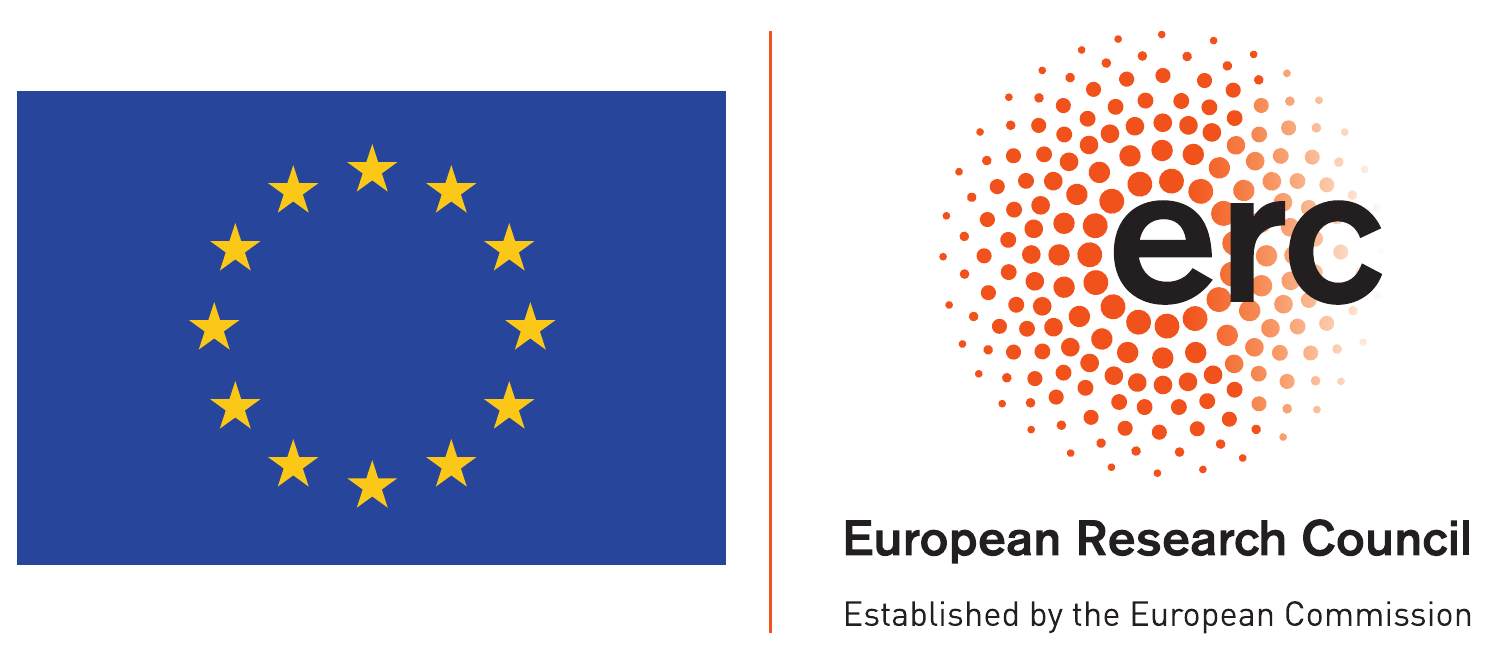}

\bibliographystyle{splncs04}
\bibliography{biblio}

\begin{thebibliography}{10}
\providecommand{\url}[1]{\texttt{#1}}
\providecommand{\urlprefix}{URL }
\providecommand{\doi}[1]{https://doi.org/#1}

\bibitem{AminiHP09}
Amini, O., Huc, F., P\'{e}rennes, S.: On the path-width of planar graphs. SIAM
  Journal on Discrete Mathematics  \textbf{23}(3),  1311--1316 (2009).
  \doi{10.1137/060670146}

\bibitem{DBLP:books/ph/BattistaETT99}
Battista, G.D., Eades, P., Tamassia, R., Tollis, I.G.: Graph Drawing:
  Algorithms for the Visualization of Graphs. Prentice Hall PTR, USA, 1st edn.
  (1998)

\bibitem{DBLP:journals/algorithmica/BertolazziBLM94}
Bertolazzi, P., {Di Battista}, G., Liotta, G., Mannino, C.: Upward drawings of
  triconnected digraphs. Algorithmica  \textbf{12}(6),  476--497 (1994).
  \doi{10.1007/BF01188716}

\bibitem{DBLP:journals/siamcomp/BertolazziBMT98}
Bertolazzi, P., {Di Battista}, G., Mannino, C., Tamassia, R.: Optimal upward
  planarity testing of single-source digraphs. {SIAM} J. Comput.
  \textbf{27}(1),  132--169 (1998). \doi{10.1137/S0097539794279626}

\bibitem{Biedl16}
Biedl, T.: Triangulating planar graphs while keeping the pathwidth small. In:
  Mayr, E.W. (ed.) Graph-Theoretic Concepts in Computer Science. pp. 425--439.
  Springer Berlin Heidelberg, Berlin, Heidelberg (2016).
  \doi{10.1007/978-3-662-53174-7_30}

\bibitem{DBLP:conf/esa/BlasiusFR21}
Bl{\"{a}}sius, T., Fink, S.D., Rutter, I.: Synchronized planarity with
  applications to constrained planarity problems. In: Mutzel, P., Pagh, R.,
  Herman, G. (eds.) 29th Annual European Symposium on Algorithms, {ESA} 2021,
  September 6-8, 2021, Lisbon, Portugal (Virtual Conference). LIPIcs, vol.~204,
  pp. 19:1--19:14. Schloss Dagstuhl - Leibniz-Zentrum f{\"{u}}r Informatik
  (2021). \doi{10.4230/LIPIcs.ESA.2021.19}

\bibitem{DBLP:journals/talg/BlasiusR16}
Bl{\"{a}}sius, T., Rutter, I.: Simultaneous pq-ordering with applications to
  constrained embedding problems. {ACM} Trans. Algorithms  \textbf{12}(2),
  16:1--16:46 (2016). \doi{10.1145/2738054}

\bibitem{Bodlaender98}
Bodlaender, H.L.: A partial \emph{k}-arboretum of graphs with bounded
  treewidth. Theor. Comput. Sci.  \textbf{209}(1-2),  1--45 (1998).
  \doi{10.1016/S0304-3975(97)00228-4}

\bibitem{BodlaenderCW22}
Bodlaender, H.L., Cornelissen, G., van~der Wegen, M.: Problems hard for
  treewidth but easy for stable gonality. In: Bekos, M.A., Kaufmann, M. (eds.)
  Graph-Theoretic Concepts in Computer Science - 48th International Workshop,
  {WG} 2022, T{\"{u}}bingen, Germany, June 22-24, 2022, Revised Selected
  Papers. Lecture Notes in Computer Science, vol. 13453, pp. 84--97. Springer
  (2022). \doi{10.1007/978-3-031-15914-5\_7}

\bibitem{BodlaenderGNS21}
Bodlaender, H.L., Groenland, C., Nederlof, J., Swennenhuis, C.M.F.:
  Parameterized problems complete for nondeterministic {FPT} time and
  logarithmic space. In: 62nd {IEEE} Annual Symposium on Foundations of
  Computer Science, {FOCS} 2021, Denver, CO, USA, February 7-10, 2022. pp.
  193--204. {IEEE} (2021). \doi{10.1109/FOCS52979.2021.00027}

\bibitem{DBLP:conf/gd/BrucknerHR19}
Br{\"{u}}ckner, G., Himmel, M., Rutter, I.: An spqr-tree-like embedding
  representation for upward planarity. In: Archambault, D., T{\'{o}}th, C.D.
  (eds.) Graph Drawing and Network Visualization - 27th International
  Symposium, {GD} 2019, Prague, Czech Republic, September 17-20, 2019,
  Proceedings. Lecture Notes in Computer Science, vol. 11904, pp. 517--531.
  Springer (2019). \doi{10.1007/978-3-030-35802-0\_39}

\bibitem{10.1007/978-3-540-30140-0_16}
Chan, H.Y.: A parameterized algorithm for upward planarity testing. In: Albers,
  S., Radzik, T. (eds.) Algorithms - {ESA} 2004, 12th Annual European
  Symposium, Bergen, Norway, September 14-17, 2004, Proceedings. Lecture Notes
  in Computer Science, vol.~3221, pp. 157--168. Springer (2004).
  \doi{10.1007/978-3-540-30140-0\_16}

\bibitem{DBLP:conf/compgeom/ChaplickGFGRS22}
Chaplick, S., {Di Giacomo}, E., Frati, F., Ganian, R., Raftopoulou, C.N.,
  Simonov, K.: Parameterized algorithms for upward planarity. In: Goaoc, X.,
  Kerber, M. (eds.) 38th International Symposium on Computational Geometry,
  SoCG 2022, June 7-10, 2022, Berlin, Germany. LIPIcs, vol.~224, pp.
  26:1--26:16. Schloss Dagstuhl - Leibniz-Zentrum f{\"{u}}r Informatik (2022).
  \doi{10.4230/LIPIcs.SoCG.2022.26}

\bibitem{CyganFKLMPPS15}
Cygan, M., Fomin, F.V., Kowalik, L., Lokshtanov, D., Marx, D., Pilipczuk, M.,
  Pilipczuk, M., Saurabh, S.: Parameterized Algorithms. Springer (2015).
  \doi{10.1007/978-3-319-21275-3}

\bibitem{DBLP:journals/siamcomp/BattistaLV98}
{Di Battista}, G., Liotta, G., Vargiu, F.: Spirality and optimal orthogonal
  drawings. {SIAM} J. Comput.  \textbf{27}(6),  1764--1811 (1998).
  \doi{10.1137/S0097539794262847}

\bibitem{BattistaT88}
Di~Battista, G., Tamassia, R.: Algorithms for plane representations of acyclic
  digraphs. Theor. Comput. Sci.  \textbf{61}(2–3),  175–198 (nov 1988).
  \doi{10.1016/0304-3975(88)90123-5}

\bibitem{DBLP:journals/jcss/GiacomoLM22}
{Di Giacomo}, E., Liotta, G., Montecchiani, F.: Orthogonal planarity testing of
  bounded treewidth graphs. J. Comput. Syst. Sci.  \textbf{125},  129--148
  (2022). \doi{10.1016/j.jcss.2021.11.004}

\bibitem{DBLP:journals/siamdm/DidimoGL09}
Didimo, W., Giordano, F., Liotta, G.: Upward spirality and upward planarity
  testing. {SIAM} J. Discret. Math.  \textbf{23}(4),  1842--1899 (2009).
  \doi{10.1137/070696854}

\bibitem{DBLP:conf/gd/DidimoKLO22}
Didimo, W., Kaufmann, M., Liotta, G., Ortali, G.: Rectilinear planarity of
  partial 2-trees. In: Angelini, P., von Hanxleden, R. (eds.) Graph Drawing and
  Network Visualization - 30th International Symposium, {GD} 2022, Tokyo,
  Japan, September 13-16, 2022, Revised Selected Papers. Lecture Notes in
  Computer Science, vol. 13764, pp. 157--172. Springer (2022).
  \doi{10.1007/978-3-031-22203-0\_12}

\bibitem{DBLP:conf/soda/DidimoLOP20}
Didimo, W., Liotta, G., Ortali, G., Patrignani, M.: Optimal orthogonal drawings
  of planar 3-graphs in linear time. In: Chawla, S. (ed.) Proceedings of the
  2020 {ACM-SIAM} Symposium on Discrete Algorithms, {SODA} 2020, Salt Lake
  City, UT, USA, January 5-8, 2020. pp. 806--825. {SIAM} (2020).
  \doi{10.1137/1.9781611975994.49}

\bibitem{DidimoLP19}
Didimo, W., Liotta, G., Patrignani, M.: Hv-planarity: Algorithms and
  complexity. Journal of Computer and System Sciences  \textbf{99},  72--90
  (2019). \doi{10.1016/j.jcss.2018.08.003}

\bibitem{DowneyF13}
Downey, R.G., Fellows, M.R.: Fundamentals of Parameterized Complexity. Texts in
  Computer Science, Springer (2013). \doi{10.1007/978-1-4471-5559-1}

\bibitem{DBLP:journals/comgeo/Frati22}
Frati, F.: Planar rectilinear drawings of outerplanar graphs in linear time.
  Comput. Geom.  \textbf{103},  101854 (2022).
  \doi{10.1016/j.comgeo.2021.101854}

\bibitem{DBLP:journals/jacm/FulekT22}
Fulek, R., T{\'{o}}th, C.D.: Atomic embeddability, clustered planarity, and
  thickenability. J. {ACM}  \textbf{69}(2),  13:1--13:34 (2022).
  \doi{10.1145/3502264}

\bibitem{GargT01}
Garg, A., Tamassia, R.: On the computational complexity of upward and
  rectilinear planarity testing. {SIAM} J. Comput.  \textbf{31}(2),  601--625
  (2001). \doi{10.1137/S0097539794277123}

\bibitem{DBLP:journals/algorithmica/GrigorievB07}
Grigoriev, A., Bodlaender, H.L.: Algorithms for graphs embeddable with few
  crossings per edge. Algorithmica  \textbf{49}(1),  1--11 (2007).
  \doi{10.1007/s00453-007-0010-x}

\bibitem{DBLP:journals/jacm/HopcroftT74}
Hopcroft, J.E., Tarjan, R.E.: Efficient planarity testing. J. {ACM}
  \textbf{21}(4),  549--568 (1974). \doi{10.1145/321850.321852}

\bibitem{DBLP:journals/siamcomp/HuttonL96}
Hutton, M.D., Lubiw, A.: Upward planning of single-source acyclic digraphs.
  {SIAM} J. Comput.  \textbf{25}(2),  291--311 (1996).
  \doi{10.1137/S0097539792235906}

\bibitem{ImpagliazzoP99}
Impagliazzo, R., Paturi, R.: Complexity of \mbox{$k$}-{SAT}. In: Proceedings of
  the 14th Annual {IEEE} Conference on Computational Complexity, Atlanta,
  Georgia, USA, May 4-6, 1999. pp. 237--240. {IEEE} Computer Society (1999).
  \doi{10.1109/CCC.1999.766282}

\bibitem{DBLP:conf/dagstuhl/1999dg}
Kaufmann, M., Wagner, D. (eds.): Drawing Graphs, Methods and Models (the book
  grow out of a Dagstuhl Seminar, April 1999), Lecture Notes in Computer
  Science, vol.~2025. Springer (2001). \doi{10.1007/3-540-44969-8}

\bibitem{KM-2012}
Korzhik, V.P., Mohar, B.: Minimal obstructions for 1-immersions and hardness of
  1-planarity testing. J. Graph Theory  \textbf{72}(1),  30--71 (2013).
  \doi{10.1002/jgt.21630}

\bibitem{LempelEC67}
Lempel, A., Even, S., Cederbaum, I.: An algorithm for planarity testing of
  graphs. In: Theory of Graphs (Internat. Sympos., Rome, 1966). pp. 215--232.
  New York: Gordon and Breach (1967),
  \url{https://cir.nii.ac.jp/crid/1574231874332098176}

\bibitem{DBLP:journals/jcss/LiottaRT23}
Liotta, G., Rutter, I., Tappini, A.: Parameterized complexity of graph
  planarity with restricted cyclic orders. J. Comput. Syst. Sci.  \textbf{135},
   125--144 (2023). \doi{10.1016/j.jcss.2023.02.007}

\bibitem{DBLP:books/ws/NishizekiR04}
Nishizeki, T., Rahman, M.S.: Planar Graph Drawing, Lecture Notes Series on
  Computing, vol.~12. World Scientific (2004). \doi{10.1142/5648}

\bibitem{DBLP:reference/crc/Patrignani13}
Patrignani, M.: Planarity testing and embedding. In: Tamassia, R. (ed.)
  Handbook on Graph Drawing and Visualization, pp. 1--42. Chapman and Hall/CRC
  (2013),
  \url{https://cs.brown.edu/people/rtamassi/gdhandbook/chapters/planarity.pdf}

\bibitem{DBLP:journals/siamcomp/Tamassia87}
Tamassia, R.: On embedding a graph in the grid with the minimum number of
  bends. {SIAM} J. Comput.  \textbf{16}(3),  421--444 (1987).
  \doi{10.1137/0216030}

\bibitem{DBLP:reference/crc/2013gd}
Tamassia, R. (ed.): Handbook on Graph Drawing and Visualization. Chapman and
  Hall/CRC (2013),
  \url{https://www.crcpress.com/Handbook-of-Graph-Drawing-and-Visualization/Tamassia/9781584884125}

\bibitem{DBLP:journals/ipl/UrschelW21}
Urschel, J.C., Wellens, J.: Testing gap k-planarity is np-complete. Inf.
  Process. Lett.  \textbf{169},  106083 (2021). \doi{10.1016/j.ipl.2020.106083}

\bibitem{VijayanW85}
Vijayan, G., Wigderson, A.: Rectilinear graphs and their embeddings. SIAM
  Journal on Computing  \textbf{14}(2),  355--372 (1985). \doi{10.1137/0214027}

\end{thebibliography}

\clearpage

\appendix
\section{Preliminaries}
%\todo[inline]{Generally ready, maybe cut the planar/upward planar part somewhat}
\subsection{Parameterized complexity and ETH} 

In parameterized complexity~\cite{CyganFKLMPPS15,DowneyF13},
the complexity of a problem is studied not only with respect to the input size, but also with respect to some problem parameter(s). The core idea behind parameterized complexity is that the combinatorial explosion resulting from the \NP-hardness of a problem can sometimes be confined to certain structural parameters that are small in practical settings. We now proceed to the formal definitions.

A \emph{parameterized problem} $\Pi$ is a subset of $\Sigma^*\times \mathbb{N}$, where $\Sigma$ is a finite alphabet. Thus, an instance of $\Pi$ is a pair $(I,k)$, where $I\subseteq\Sigma^*$ and $k$ is a nonnegative integer called a \emph{parameter}. A parameterized problem $\Pi$ is \emph{fixed-parameter tractable} (\FPT) if it can be solved in $f(k)\cdot |I|^{\Oh(1)}$ time for some computable function $f(\cdot)$. 
Parameterized complexity theory also provides tools to refute the existence of an \FPT~algorithm for a parameterized problem. The standard way is to show that the considered problem is hard in the parameterized complexity classes \W[1] or \W[2]. We refer to the book~\cite{CyganFKLMPPS15} for the formal definitions of the parameterized complexity classes.  The basic complexity assumption  of the theory is that for the class  \FPT, formed by all parameterized fixed-parameter tractable problems, $\FPT\subset \W[1]\subset \W[2]$. The hardness is proved by demonstrating a parameterized reduction from a problem known to be hard in  the considered complexity class.  
A \emph{parameterized reduction} is a many-one reduction that takes an input $(I,k)$ of the first problem, and in $f(k)|I|^{\Oh(1)}$ time outputs an equivalent instance $(I',k')$ of the second problem with $k'\leq g(k)$, where $f(\cdot)$ and $g(\cdot)$ are computable functions. Another way to obtain lower bounds is to use the \emph{Exponential Time Hypothesis (ETH)} formulated by Impagliazzo, Paturi and Zane~\cite{ImpagliazzoP99}. For an integer $q\geq 3$, let $\delta_q$ be the infimum of the real numbers $c$ such that the \textsc{$q$-CNF-Satisfiability} problem can be solved in time $\Oh(2^{c n})$, where $n$ is the number of variables. The Exponential Time Hypothesis states that $\delta_3>0$. In particular, ETH implies that \textsc{$3$-CNF-Satisfiability} cannot be solved in time $2^{o(n)}n^{\Oh(1)}$.

\subsection{Treewidth and Pathwidth}
	A \emph{tree decomposition}~$\mathcal{T}$ of a graph $G$ is a pair $(T,\chi)$, where $T$ is a tree (whose vertices we call \emph{nodes}) rooted at a node $r$ and $\chi$ is a function that assigns to each node $t \in \mathcal{T}$ a set $\chi(t) \subseteq V(G)$ such that the following holds: 
\begin{itemize}
	\item For every $\{u, v\} \in E(G)$ there is a node	$t$ such that $u,v\in \chi(t)$.
	\item For every vertex $v \in V(G)$, the set of nodes $t$ satisfying $v\in \chi(t)$ forms a nonempty subtree of~$T$.
\end{itemize}

 The sets $\chi(t)$, for $t \in V(T)$, are called \emph{bags} of the tree decomposition.
The \emph{width} of a tree decomposition $(T,\chi)$ is the size of a largest set $\chi(t)$ minus~$1$, and the \emph{treewidth} of the graph $G$,
denoted $\tw(G)$, is the minimum width of a tree decomposition of~$G$.

The pathwidth is defined similarly in terms of paths. A \emph{path decomposition}~$\mathcal{P}$ of a graph $G$ is a pair 
$(P,\chi)$, where $P$ is a path and $\chi$ is a function that assigns to each node $p \in P$ a set $\chi(p) \subseteq V(G)$ such that the following holds: 
\begin{itemize}
	\item For every $\{u, v\} \in E(G)$ there is a node	$p$ such that $u,v\in \chi(p)$.
	\item For every vertex $v \in V(G)$, the set of nodes $p$ satisfying $v\in \chi(p)$ forms a nonempty subpath of~$P$.
\end{itemize}

 The sets $\chi(p)$, for $p \in V(P)$, are called \emph{bags} of the path decomposition.
The \emph{width} of a path decomposition $(P,\chi)$ is the size of a largest set $\chi(p)$ minus~$1$, and the \emph{pathwidth} of the graph $G$,
denoted $\pw(G)$, is the minimum width of a path decomposition of~$G$. 

Throughout the paper, whenever we consider the treewidth or pathwidth of a directed graph, we mean the respective parameter of the underlying graph.
 
It is known that the pathwidth of the dual graph is bounded in terms of the pathwidth of the primal.
\begin{theorem}[Amini, Huc, and P\'{e}rennes~\cite{AminiHP09}]
    \label{theorem:pw_dual}
    For a triconnected planar graph $G$, $\pw(G^*) \le 3\pw(G) + 2$, where $G^*$ is the dual graph of $G$.
\end{theorem}

Note that this result is constructive when a path decomposition is given. We also recall that triangulating the graph increases the pathwidth by at most a constant factor.

\begin{theorem}[Biedl~\cite{Biedl16}]
    \label{theorem:pw_triangulation}
    There is a polynomial-time algorithm that, given a simple planar graph $G$ of pathwidth~$k$ on at least three vertices, outputs a plane triangulation~$G'$ of~$G$ such that~$\pw(G') \in \Oh(k)$.
\end{theorem}

In our reductions, we construct the target graph by replacing the edges of an intermediate graph with certain gadget graphs. For completeness, we prove in the next lemma the folklore bound on the pathwidth of the resulting graph.

\begin{lemma}\label{lemma:pw_gadgets}
    Let $G$ be a graph and $H_1$, \ldots, $H_t$ be a family of graphs, where for each $i \in [t]$, there are two distinct special vertices $s_i$, $t_i \in V(H_i)$. Consider a graph $F$ obtained from $G$ by replacing every edge $\{u,v\} \in E(G)$ with a copy of one of the graphs $H_i$, $i \in [t]$, such that $u$ becomes associated with the vertex $s_i$ of the copy, and $v$ with $t_i$. Then
    $\pw(F) \le \pw(G) + \max_{i \in [t]} \pw(H_i) + 1$.
\end{lemma}
\begin{proof}
    Consider the optimal path decomposition $(P, \chi)$ of $G$. For every edge $\{u,v\} \in E(G)$ perform the following operation on $(P, \chi)$: Take an arbitrary $p \in V(P)$, such that $u, v \in \chi(p)$. Let $H_i$ be the replacement of $\{u,v\}$ in $F$, and consider the optimal path decomposition of $H_i$, where the set $\chi(p)$ is added to every bag, and the occurences of $s_i$ ($t_i$) are removed. Now make a copy of the bag $\chi(p)$, and insert the decomposition of $H_i$ augmented as above between the two copies. Let $(P', \chi')$ be the resulting decomposition, we claim that it is a path decomposition of $F$ of desired width.

    By construction, every bag of the new decomposition is a subset of a union of a bag of $P$ and a bag in the optimal path decomposition of $H_i$ for some $i \in [t]$, thus the size of any bag is at most $\pw(G) + 1 + \max_{i \in [t]} \pw(H_i) + 1$. For correctness, first consider an internal vertex of some $H_i$ in $F$ (i.e., a vertex that is not $s_i$ or $t_i$ in that $H_i$). Its appearances are exactly as in the optimal path decomposition of $H_i$, thus his vertex indeed forms a subpath in $(P', \chi')$. Consider a vertex of $F$ that originates from a vertex of $G$, its appearances changed from a subpath of $P$ to a subpath of $P'$, since every newly inserted bag contains fully the bag $\chi(p)$ it was based on. Finally, every edge of every inserted $H_i$ is present in $(P', \chi')$ by the inserted augmented decomposition of $H_i$.
\end{proof}
%\bmpr{I removed the 'at least 3 vertices' requirement since the use of Big-oh makes it trivially true for graphs with 1 or 2 vertices.}

%We recall that the treewidth of \emph{embedding graph}~\cite{ChaplickGFGRS22}, which intuitively speaking captures the relation of vertices with their incident edges and faces, is bounded linearly in terms of treewidth of the original graph.
%For a formal definition, let $G$ be a connected graph with a planar embedding $\Epsilon$, and let $F$ be the set of faces of $\Epsilon$. Let $G^-$ be the graph obtained from $G$ by subdividing each edge $e$ once, creating the vertex $v_e$. The embedding graph $\widetilde{G}$ of $G$ is the graph obtained from $G^-$ by adding a vertex $f$ for each face in $F$, and connecting $f$ to each vertex in $G^-$ incident to $f$. Formally, $V(\widetilde{G})=V(G)\cup F \cup \{v_e~|~e\in E(G)\}=\Vv\cup \Vf\cup \Ve$, and $E(\widetilde{G})=E(G^-)\cup  \{fv~|~v\in V(G) \wedge v\text{ is incident to f }\} \cup \{fv_e~|~e\in E(G) \wedge \text{ both endpoints of $e$ are incident to f }\}$. Observe that $\widetilde{G}$ is tripartite, and the three sets of vertices that occur in the definition of $V(\widetilde{G})$ are called the \emph{true vertices}, \emph{face-vertices} and \emph{edge-vertices} of $\widetilde{G}$, respectively.

%\begin{theorem}[Theorem 17, \cite{ChaplickGFGRS22}]
%\label{thm:embeddingtw}
%Let $G$ be a graph with a planar embedding of treewidth $k$ where $k\geq 1$. Then the embedding graph $\widetilde{G}$ has treewidth at most $11k-4\in \Oh(k)$.
%\end{theorem}

\subsection{Planar drawings and embeddings} \label{subse:planar-definitions}

	A \emph{drawing} of a graph maps each vertex to a point in the plane and each edge to a Jordan arc between the end-points of the edge. A drawing is \emph{planar} if no two edges intersect, except at common end-points. A planar drawing partitions the plane into regions, called \emph{faces}. The bounded faces are called \emph{internal}, while the unbounded face is the \emph{outer face}. 
	Two planar drawings of a graph are \emph{equivalent} if: (1) they have the same \emph{rotation system}, that is, for each vertex $v$, the clockwise order of the edges incident to $v$ is the same in both drawings; and (2) their outer faces are delimited by the same walk, that is, the order of the edges encountered when clockwise traversing the boundary of the outer face is the same in both drawings. A \emph{planar embedding} of a graph is an equivalence class of planar drawings of that graph. 
	
	Thus, a planar embedding of a graph consists of a rotation system and a choice for the walk delimiting the outer face. We often talk about a \emph{face of a planar embedding}, meaning a face of any planar drawing that respects the planar embedding. The \emph{flip} of a planar embedding is the planar embedding obtained by reversing the clockwise order of the edges incident to each vertex and by reversing the order of the edges encountered when clockwise traversing the boundary of the outer face.
	
	%A \emph{planar graph} is a graph that admits a planar drawing. 
	%A graph together with a specified planar embedding is a \emph{plane graph}.	

\subsection{Upward planar drawings and embeddings} \label{subse:upward-definitions}

	Throughout the paper, we use the term \emph{digraph} as short for ``directed graph''. A digraph is \emph{acyclic} if it contains no directed cycle. An acyclic digraph is usually called \emph{DAG}, for short. A vertex in a digraph is a \emph{source} if it is only incident to outgoing edges and it is a \emph{sink} if it is only incident to incoming edges. For an edge $uv$  (oriented from $u$ to $v$) of a digraph, $u$ is the \emph{source of $uv$} and $v$ is the \emph{sink of $uv$}.
	A vertex in a digraph is a \emph{switch} if it is a source or a sink, and it is a \emph{non-switch} otherwise. The \emph{underlying graph} of a digraph is the undirected graph obtained from the digraph by ignoring the edge directions. 

	A \emph{plane digraph} is a digraph together with a prescribed planar embedding for its underlying graph.

	A drawing of a digraph is \emph{upward} if every edge is represented by a Jordan arc which is monotonically increasing in the $y$-direction from the source to the sink of the edge, and it is \emph{upward planar} if it is both upward and planar. A digraph is \emph{upward planar} if it admits an upward planar drawing; we use \textsc{Upward Planarity Testing} to denote the problem of determining whether a digraph is upward planar; w.l.o.g., we assume that the input digraph is connected. 

	Consider an upward planar drawing $\Gamma$ of a digraph $G$. An \emph{angle} $\alpha$ of a face $f$ of $\Gamma$ is a triple $(e_1,v,e_2)$, where $e_1$ and $e_2$ are two edges of $G$ that are incident to the vertex $v$, that are incident to the face $f$, and that are consecutive in the order of the edges encountered when clockwise traversing the boundary of $f$. We say that $\alpha$ is \emph{flat} if one between $e_1$ and $e_2$ is incoming $v$ and the other one is outgoing $v$, otherwise $\alpha$ is a \emph{switch angle}.
	Then $\Gamma$ defines an \emph{angle assignment}, which assigns the value $-1$, $0$, and $1$ to each small, flat, and large angle, respectively, in every face of $\Gamma$. The angle assignment, together with the planar embedding of the underlying graph of $G$ in $\Gamma$, constitutes an \emph{upward planar embedding} of~$G$.

	A switch angle at a vertex $v$ is hence delimited by two outgoing or by two incoming edges for $v$. Each switch angle is further classified as \emph{large} or \emph{small} as follows. Consider a switch angle $\alpha=(e_1,v,e_2)$ at a vertex $v$ in a face $f$ delimited by two outgoing edges (resp.\ by two incoming edges) and consider a disk $D$ centered at $v$, sufficiently small so that its boundary has a single intersection with every edge incident to $v$. The edges $e_1$ and $e_2$ divide $D$ into two regions, one of which contains part of $f$ and contains no portion of any edge incident to $v$ in its interior; call $D'$ this region. Then we say that $\alpha$ is \emph{large} if $D'$ contains a suitably short vertical segment that has $v$ as its highest (resp.\ lowest) end-point, it is \emph{small} otherwise.

	An upward planar drawing hence defines an \emph{angle assignment}, which is an assignment of the value $-1$, $0$, and $1$ to each small, flat, and large angle, respectively, in every face of $\Gamma$. This angle assignment, together with the planar embedding of the underlying graph of $G$ in $\Gamma$, constitutes an \emph{upward planar embedding} of $G$. An \emph{upward plane digraph} is a digraph together with a prescribed upward planar embedding.
	
	%An upward planar embedding is hence an equivalence class of upward planar drawings, where two drawings are equivalent if: (1) they have the same rotation system; (2) their outer faces are delimited by the same walk; and (3) they induce the same angle assignment. 

	The angle assignments that enhance a planar embedding into an upward planar embedding have been characterized by Didimo et al.~\cite{DBLP:journals/siamdm/DidimoGL09}, building on the work by Bertolazzi et al.~\cite{DBLP:journals/algorithmica/BertolazziBLM94}. Note that, once the planar embedding $\Epsilon$ of a digraph $G$ is specified, then so are the angles of the faces of $\mathcal E$; in particular, whether an angle is flat or switch only depends on $\mathcal E$. Consider an angle assignment for $\Epsilon$. If $v$ is a vertex of $G$, we denote by $n_i(v)$ the number of angles at $v$ that are labeled $i$, with $i \in \{-1,0,1\}$. If $f$ is a face in $\Epsilon$, we denote by $n_i(f)$ the number of angles of $f$ that are labeled $i$, with $i \in \{-1,0,1\}$. The cited characterization is as follows.
		
	%This characterization allows one to formulate the upward planarity testing problem as the problem of testing for the existence of a planar embedding and of an assignment for the angles in the faces such that certain combinatorial properties are satisfied. 
	
\begin{theorem}[\cite{DBLP:journals/algorithmica/BertolazziBLM94,DBLP:journals/siamdm/DidimoGL09}]\label{th:upward-conditions}
		Let $G$ be a digraph, $\Epsilon$ be a planar embedding of the underlying graph of $G$, and $\lambda$ be an assignment of each angle of each face in $\Epsilon$ to a value in $\{-1,0,1\}$. Then $\Epsilon$ and $\lambda$ define an upward planar embedding of $G$ if and only if the following properties hold:
		\begin{description}
			\item[UP0] If $\alpha$ is a switch angle, then $\lambda(\alpha)\in\{-1,1\}$, and if $\alpha$ is a flat angle, then $\lambda(\alpha)=0$.
			\item[UP1] If $v$ is a switch vertex of $G$, then $n_1(v)=1$, $n_{-1}(v)=\deg(v)-1$, ${n_0(v)=0}$.
			\item[UP2] If $v$ is a non-switch vertex of $G$, then $n_1(v)=0$, $n_{-1}(v)=\deg(v)-2$, $n_0(v)=2$.
			\item[UP3] If $f$ is a face of $G$, then $n_1(f)=n_{-1}(f)-2$ if $f$ is an internal face, and $n_{1}(f)=n_{-1}(f)+2$ if $f$ is the outer face.
		\end{description}
	\end{theorem}

Following condition \textbf{UP3}, we say that the \emph{contribution} of a subpath of the boundary of a face to this face is the sum of the values assigned to the angles at internal vertices of the subpath towards the face.

Related to upward planarity is the notion of $st$-planar graphs. An \emph{$st$-planar graph} is an acyclic orientation of a planar graph such that there is exactly one source and one sink, and the associated planar embedding where the source and the sink lie on the outer face. It is known that an acyclic digraph is upward planar if and only if it is a subgraph of an $st$-planar graph~\cite{BattistaT88}. Moreover, a triconnected $st$-planar graph has a unique upward planar embedding as the planar embedding together with the choice of the outer face is fixed, and, consequently, the angle assignments is fixed too.

In our hardness reduction, we use the following gadgets called \emph{tendrils}, introduced by Garg and Tamassia~\cite{GargT01}.
 \begin{lemma}[\cite{GargT01}]
     \label{lemma:tendrils}
     For every $k \ge 0$, there exists a directed acyclic graph $T_k$ with two special vertices called \emph{poles}, such that the following properties hold:
     \begin{itemize}
         \item one of the poles is a source and the other is a sink in $T_k$,
         \item after adding the edge connecting the poles, $T_k$ becomes triconnected,
         \item $T_k$ admits a unique upward planar embeddding,
         \item under this embedding, the angle contribution of the two boundary paths from the source pole to the sink pole are $2k$ and $-2k$, respectively,
         \item $T_k$ is of size $\Oh(k)$ and of pathwidth 2.
     \end{itemize}
 \end{lemma}
 %In the original paper the tendrils $T_k$ have contribution $2k$ and $-2k$ respectively; to achieve $6k$ and $-6k$ we simply scale up the indices.

\subsection{Rectilinear planarity} \label{subse:rect-definitions}
It is also well-known that rectilinear embeddings can be characterized by angle assignments under a planar embedding. Let $G$ be an undirected graph of degree at most 4, and let $\Epsilon$ be its planar embedding. Consider an angle assignment that maps every angle in $\Epsilon$ to a value in \{1, 2, 3, 4\}. For a face $f$ in $\Epsilon$, we denote by $n_i(f)$ the number of angles of $f$ that are labeled $i$, with $i \in \{1, 2, 3, 4\}$. The characterization is as follows.

\begin{theorem}[\cite{DBLP:journals/siamcomp/Tamassia87,VijayanW85}]\label{th:rect-conditions}
		Let $G$ be a graph, $\Epsilon$ be a planar embedding of $G$, and $\lambda$ be an assignment of each angle of each face in $\Epsilon$ to a value in $\{1, 2, 3, 4\}$. Then $\Epsilon$ and $\lambda$ define a rectilinear embedding of $G$ if and only if the following properties hold:
		\begin{description}
			\item[RE0]\label[condition]{th:rect-conditions-vertex} For every vertex $v \in V(G)$, the sum of the labels of angles around $v$ is $4$.
			\item[RE1]\label[condition]{th:rect-conditions-face} For every face $f$ in $\Epsilon$,
   $$2 \cdot n_4(f) + n_3(f) - n_1(f) = \begin{cases}
   -4\quad\text{if $f$ is an internal face},\\
   4\quad\text{if $f$ is the outer face}.
   \end{cases}$$
		\end{description}
	\end{theorem}

 In the light of condition RE1, we say that the \emph{contribution} of a subpath of the boundary of a face to this face is twice the number of fours plus the number of threes minus the number of ones among the values assigned to the angles at internal vertices of the subpath towards the face.
 Similarly to upward planarity, we use the following gadgets called \emph{rectilinear tendrils}, introduced by Garg and Tamassia~\cite{GargT01}.
 \begin{lemma}[\cite{GargT01}]
     \label{lemma:rect_tendrils}
     For every $k \ge 0$, there exists an undirected graph $T_k$ with two special vertices called \emph{poles}, such that the following properties hold:
     \begin{itemize}
         \item $T_k$ is rectilinear planar,
         \item the poles have degree one in $T_k$,
         \item $T_k$ admits exactly four rectilinear embeddings that all share the same underlying planar embedding where the poles are both on the outer face,
         \item under this embeddings, the angle contribution of the two boundary paths from the source pole to the sink pole are $f$ and $-f$, respectively, where $f \in \{4k, 4k + 1, 4k + 2\}$,
         \item $T_k$ is of size $\Oh(k)$ and of pathwidth 2.
     \end{itemize}
 \end{lemma}
 We say that the tendril $T_k$ has a \emph{significant contribution} $4k$ to a face if its contribution is in $\{4k, 4k + 1, 4k + 2\}$, and $-4k$ otherwise.

\section{From \textsc{Multicolored Clique} to \textsc{All-or-Nothing-Flow}}
In this section, we present an FPT-reduction that given an instance $(G, k)$ of \mclique,  constructs an equivalent
instance of \AoNF on a planar graph of pathwidth~$\Oh(k)$. We present several useful gadgets, before describing the construction of the flow instance.

%\defparquestion{\mclique}{An undirected graph~$G$ and a partition of its vertex set into~$k$ independent sets~$V_1, \ldots, V_k$, each consisting of~$N$ vertices.}{$k$.}{Does~$G$ contain a clique~$C \subseteq V(G)$ such that~$|C \cap V_i| = 1$ for each~$i \in [k]$?}

%\defquestion{\AoNF}{A flow network~$(G,c,s,t)$ and a positive integer~$\calF$.}{Does there exist an $st$-flow of value exactly~$\calF$, such that the flow through any arc~$uv \in E(G)$ is either~$0$ or equal to~$c(uv)$?}

%\defquestion
%{\COr}
%{An undirected graph~$G$ with an edge-capacity function~$c_H \colon E(G) \to \mathbb{Z}_{\geq 0}$.}
%{Is it possible to orient the edges of~$G$, such that for each vertex~$v \in V(G)$ the total capacity of edges oriented into~$v$ is equal to the total capacity of edges oriented out of~$v$? (Such an orientation is called a \emph{circulating orientation}.)}

\subsection{Reduction to \AoNF}
The vertex set $V(G)$ of a \mclique instance $(G, k)$ is partitioned into~$k$ sets~$V_1, \ldots, V_k$. For each~$i \in [k]$, let $V_i=\{v_{i, 1}, \dots, v_{i, N}\}$. Then, let $\overline{E}(G)=\{\{u,v\}\mid u\in V_i, v\in V_j, \{u,v\}\not\in E(G)\}$, i.e. the set of non-edges of $G$; let $|\overline{E}(G)|=m$ and number the edges of~$\overline{E}(G)$ arbitrarily from~$1$ to~$m$. 

\paragraph{Gadgets.}
Now, to facilitate the presentation, we describe a \emph{vertex selection gadget} ($\VS$ gadget) for a vertex set~$V_i=\{v_{i, 1}, \dots, v_{i, N}\}$ and $j\in \mathbb{N}$ ($\VS_i^j$,~\cref{fig:VS_gadget}).
\begin{itemize}
    \item First, we introduce a pair of vertices $V_i^j$ and $V_i^{j+1}$, and a vertex set $\left\{v_{i, {q}}^{j}, u_{i, {q}}^{j}, w_{i, {q}}^{j}, g_{i, {q}}^{j}, h_{i, {q}}^{j}\right\}_{{q}\in [N]}$.
    \item For each $q\in [N]$,  we introduce an oriented path from $V_i^j$ to $V_i^{j+1}$ through the vertices $v_{i, {q}}^{j}, u_{i, {q}}^{j}, w_{i, {q}}^{j}, g_{i, {q}}^{j}, h_{i, {q}}^{j}$ in the specified order, also each arc of such a path has capacity $2kN+2q$.
\end{itemize}

%\bmpr{Naming scheme: variable $i$ to select one of the $k$ sets in the multicolor clique instance; variable $j$ to select one of the non-edges of the graph (columns), and variable $q$ to select one of the $N$ vertices within one set.}
%\lk{Seems to be ok in the current section now}

\begin{figure}[t]
    \centering
        \includegraphics[page=1]{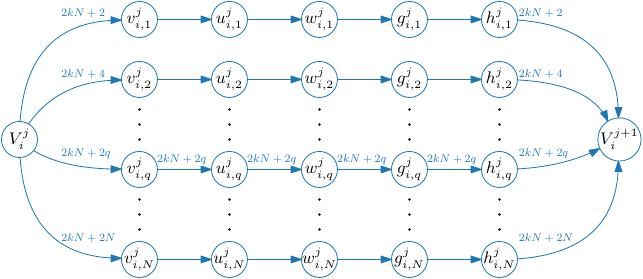} 
    \caption{The $\VS_i^j$ gadget.}
    \label{fig:VS_gadget}
\end{figure}

At this point, we can introduce the \emph{checker gadget}. The checker gadget is a specific combination of $k$ suitable $\VS$ gadgets. For each $j\in[m]$, for the $j^{\text{th}}$ non-edge $\{v_{i,a}v_{\ell,b}\}$ of~$G$ (without loss of generality, $i\leq\ell$) we construct its checker gadget $\CH^j$ as follows.
%\bmpr{I think we cannot demand~$i < \ell$ without loss of generality since we also want to have gadgets for~$i=\ell$; but we should say~$i \leq \ell$. Also: I would find it more intuitive to refer to these as 'non-edges of~$G$' rather than 'edges of~$\overline{E}(G)$'.}

\begin{itemize}
    \item For each $i\in [k]$, we introduce a $\VS_i^j$ gadget.
    \item Then,  we add a pair of vertices: $x^j$ and $y^j$.
    \item We add an oriented path of capacity $1$ going from $x^j$ to $y^j$ through $w$-vertices of all VS gadgets introduced here, as shown in~\cref{fig:CH_gadget}.
    \item We add four more oriented paths of capacity $1$, going in the following fashion: 
    \begin{itemize}
        \item it starts at $V_i^j$ ($V_{\ell}^j$), goes to $v_{i, N}^{j}$ ($u_{\ell, N}^{j}$) and then reaches $x^j$ through the $v$-vertices ($u$-vertices) of VS gadgets with a subscript less than or equal to $i$ (less than or equal to $\ell$);
        \item it starts at $y^j$, goes to the $h_{\ell, 1}^{j}$ ($g_{i, 1}^{j}$) through the $h$-vertices ($g$-vertices) of VS gadgets with a subscript greater than or equal to $\ell$ (greater than or equal to $i$), and then reaches $V_{\ell}^{j+1}$ ($V_{i}^{j+1}$).%\bmpr{Rather than referring to 'left $v$-vertices' and 'right $v$-vertices', which makes the validity of the proof dependent on the picture, I suggest using 5 different letters for the horizontal vertices so it is easier to be concrete here without relying on the picture.}
    \end{itemize}
    
    \item We decrease the capacity of the arcs following the path from $V_i^j$ to $V_i^{j+1}$ (from $V_{\ell}^j$ to $V_{\ell}^{j+1}$) through the $a^{\text{th}}$ ($b^{\text{th}}$) row of the $\text{VS}_{i}^j$ ($\text{VS}_{\ell}^j$) gadget by one,  i.e. set it to $2kN+2a-1$ ($2kN+2b-1$).
\end{itemize}

The final construction of a $\CH^j$ gadget is shown in~\cref{fig:CH_gadget}.

\begin{figure}
    \centering
        \includegraphics[page=2]{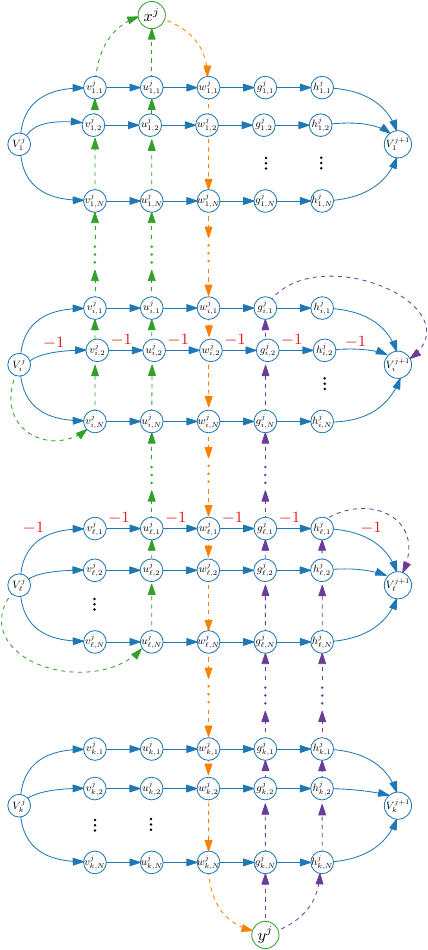} 
    \caption{The $CH^j$ gadget for an non-edge $v_{i,2} v_{\ell, 1}$. All arcs of colored paths have capacity $1$.}
    \label{fig:CH_gadget}
\end{figure}

\paragraph{The final reduction.} 
We now present the reduction. Consider an instance $(G, k)$ of \mclique. We construct a flow network~$(\calG, s, t, c)$ as follows. See also \cref{fig:multicolored-k-clique-reduction} for a schematic representation..
\begin{itemize}
    \item First, we introduce a source vertex $s$ and sink vertex $t$.
    \item Then, for each $j\in [m]$, we introduce the $\CH^j$ gadget. Note that, for $j\in [m-1]$, two VS gadgets $\VS_i^j$ and $\VS_i^{j+1}$ share the common vertex $V_i^{j+1}$.
    \item For each $i\in[k]$, we add:
    \begin{itemize}
        \item one arc of capacity $2kN$ and $N$ arcs of capacity $2$ from $s$ to $V_i^1$, and
        \item one arc of capacity $2kN$ and $N$ arcs of capacity $2$ from $V_i^{m+1}$ to $t$. 
    \end{itemize}
    \item At the end, we add a set $A$ of $k(N-1)$ arcs from $s$ to $t$, each of capacity~$2$.
\end{itemize}

This concludes the construction of $\mathcal{G}$. The reduction returns $(\mathcal{G}, \calF)$ as an instance of \AoNF,  where $\calF=k(2kN+2N)$.

\subsection{Correctness of the Reduction to \AoNF}
Suppose, given an instance $(G, k)$ of \mclique, that the reduction from the previous subsection returns $(\mathcal{G}, \calF)$ as an instance of \AoNF.
\begin{lemma}
\label{lemma:AoNF-1}
If $(G, k)$ is a yes-instance of \mclique, then $\mathcal{G}$ admits an all-or-nothing $st$-flow with value $\calF$.
\end{lemma}

\begin{proof}
Let $C$ be a multicolored $k$-clique in $G$.
For each $i\in [k]$, without loss of generality, let $v_{i, a_i}\in V(C)$ for some $a_i\in [N]$.
%\bmpr{I gave the $a$ variable a subscript~$a_i$ since it differs for each~$i$.}
Then, let the flow $f$ in $\calG$ use the arc of capacity $2kN$ and $a_i$ arcs of capacity $2$ from $s$ to $V_i^1$. 
Thus, for each $i\in [k]$, there is an inflow $f_i=2kN+2a_i$ at $V_i^1$ that is an even value in the interval from $2kN+2$ to $2kN+2N$.
Let $S=\sum_{i\in[k]}f_i$. Notice, inequality $k(2kN+2)\leq S \leq k(2kN+2N)$ holds.
Recall that~$A$ consists of $k(N-1)$ parallel $st$-arcs, each of capacity $2$. Since $\calF-S\leq 2k(N-1)$, sending the rest of the flow over a subset of the arcs of $A$ will be enough to achieve the total flow of value $\calF$ if, in turn, each of $f_i$ reaches $t$.

At this point, we are left to check whether, for each $i\in [k]$, an inflow at $V_i^1$ can be propagated to $t$ through $\calG$.
For each $i\in[k]$, by construction, there is a directed path from $V_i^1$ to $V_i^{m+1}$ whose arcs are of capacity either $2kN+2a_i$ or $2kN+2a_i-1$. Thus, for each $j\in [m]$, we go from column to column and if there is a path of  capacity $2kN+2a_i$ through the $\text{VG}_i^j$ gadget, we direct the flow through it. Otherwise, by construction there is a path of capacity $2kN+2a_i-1$ through the $\text{VG}_i^j$ gadget; also, there is a path of capacity $1$ that goes through either the $v$- or $u$-vertices up to $x^j$, which reaches $y^j$ through the $w$-vertices and then reaches $V_{i}^{j+1}$ through either the $g$- or $h$-vertices of the VS gadgets.

To route all these parts of the flow simultaneously, it is crucial that there is no column $j$ where two rows both need to use the path from $x^j$ to $y^j$ at the same time. Suppose the opposite. Then $v_{i, a_i}$ ($v_{\ell, a_{\ell}}$) propagates the flow of value $2kN+2a_i-1$ ($2kN+2a_\ell-1$) through $\text{VS}_i^j$ ($\text{VS}_{\ell}^j$), and the flow of value $1$ goes through the path $x^jy^j$. By construction, the capacities of only two paths in $\text{CH}^j$ were decreased by one, each of them corresponds to an endpoint of the $j^{\text{th}}$ non-edge. So, $\{v_{i, a_i}v_{\ell, a_\ell}\}\in \overline{E}(G)$. This contradicts the assumption that $C$ is a $k$-clique in $G$.

Hence, for each $i\in [k]$, $f_i$ reaches $V_i^{m+1}$ and afterwards goes from $V_i^{m+1}$ to $t$ through the arc of capacity $2kN$ and through $a_i$ arcs of capacity $2$. As outlined earlier, together with the flows through the arcs of $A$, the resulting value of the flow is $k(2kN+2N)$, i.e. the required $\calF$.
\end{proof}

\begin{lemma}
\label{lemma:AoNF-2}
If $\mathcal{G}$ admits an all-or-nothing $st$-flow with value $\calF$, then there is a multicolored $k$-clique in $G$.
\end{lemma}

\begin{proof}
Suppose that there is an all-or-nothing flow $f$ of value $\calF$ in the flow network $\calG$.
Since $d^+(s)=\calF$, for each $i\in [k]$, the equality $2kN<d^-(V_i^j)\leq 2kN + 2N$ holds.
The second part of the inequality is straightforward since $2kN+2N$ is the total value of all incoming edges of $V_i^1$.
To prove the first, assume, by contradiction, that for some $i\in [k]$, $d^-(V_i^1)<2kN+2$.
There are two cases, the first is for $d^-(V_i^1)>2$. But by the construction of VS gadget, at least $2kN+2$ inflow is necessary to propagate something further through one of its horizontal paths.
As for the second case, consider the following sum of inflows, i.e. $\sum_{\ell\in[k], \ell\neq i}d^-(V_{\ell}^j)\leq (k-1)(2kN+2N)$. So, we have at least $2kN+2N$ units left to reach the flow value $\calF$. 
Except the VS gadgets, $f$ could propagate the flow through the set $A$ of $k(N-1)$ arcs of the capacity $2$.
But $2kN-2k$ with even 2-inflow into $V_i^1$ is less than the necessary $2kN+2N$, a contradiction.
Thus, if the all-or-nothing flow $f$ of value $\calF$ exists, then, for each $j\in [k]$, $d^-(V_i^1)$ is an even value between $2kN+2$ and $2kN+2N$, depending on the number of incoming arcs (from $s$ of capacity $2$) used by $f$.

We claim that for each $i\in [k]$ and $j\in[m]$, $d^+(V^j_i)$ is the same as $d^-(V_i^{j+1})$, in other words, the flow cannot 'escape' from the correspondent VS gadget. Let $d^-(V^j_i)=\calF'$. Then the flow of value $\calF'$ or $\calF'-1$ is propagated through the VS gadget itself by  construction, and, if necessary for the second case, a flow of value $1$ goes through the vertices of the CH gadget up to $x^j$, but then it can only go to~$y^j$, which in turn goes to~$V^{j+1}_i$.
%This together with the fact that, for all $i\in[k]$, the $\indeg(V_i^1)=2kN+2N$, gives us an upper bound of $2kN+2N$ for the inflow of $V_i^j$ for each $i\in [k]$ and $j\in [m+1]$.
%Assume for a contradiction that flow $f$ uses two horizontal paths in $\VS_i^j$ to go from $V_{i}^j$ to $V_{i}^{j+1}$. But, for any pair of vertices $v_{i,a}, v_{i,b}$ of the same independent set $V_i$ of the initial instance $G$ there is a column in $\calG$ that corresponds to the non-existence of edge $\{v_{i,a},v_{i,b}\}$ non-existence, i.e. the flow of value $1$ is redirected for both $a^{\text{th}}$ and $b^{\text{th}}$ rows in $\VS_i^{j}$. Thus we have a contradiction, since both of these subflows go through the path $x^jy^j$ of the capacity $1$.\bmpr{I think some more details are in order here. In particular, we also rely on the manner in which we chose the capacities on the edges to ensure that the combined value of 'two small indices' are too large to fit into the larger capacity of one 'high index'.}
%This way, for each $i\in [k]$, $d^+(\VS_i^{m+1})=d^+(\VS_i^1)$.
This ensures that the flow propagates the same value through all VS gadgets on the same row.

To conclude, consider the following set $C$ of vertices in $\calG$: for each $i\in[k]$, consider the value on the inflow into $V_i^1$, let it be $d^-(V_i^1)=2kN+2a_i$; then we add $v_{i, a_i}$ from $V_i$ to the set $C$. 
According to the arguments above, such vertex is uniquely defined for each $V_i$, $i\in[k]$.
The existence of a pair $v_{i,a}, v_{\ell,b}\in C$ such that $\{v_{i,a}v_{\ell,b}\}\not\in E(G)$ contradicts the construction since, for each non-edge of $G$, there is a CH-column in $\calG$ that forbids inflows of values $2kN+2a$ in $V_i^j$ and $2kN+2b$ in $V_{\ell}^j$ simultaneously.
This completes the proof: Indeed, the vertex set $C$ is a multicolored clique in $G$.
\end{proof}

One of the key observations to obtain the result we are aiming for is the bounded pathwidth of $\calG$.

\begin{lemma} \label{lemma:aonf:pathwidth}
    The pathwidth of $\calG$ is $\Oh(k)$.
\end{lemma}
\begin{proof}
    Without loss of generality, let us fix $j\in[m]$ and consider the following decomposition of a single $\CH^j$ gadget of graph $\calG$. 
    We introduce a path $P_j=P_j^1P_j^2\dots P_j^{kN}$  and associate the vertex set $\{x^j, y^j\}\cup\{V_i^j, V_i^{j+1}\}_{i\in [k]}$ with the bag (node) $P_j$. Now, for each $1<\ell\leq [kN]$, let $p=\lceil \ell/N \rceil$, $q=\{\ell/N \}$; we add $\left\{v_{p, q}^{j}, u_{p, q}^{j}, w_{p, q}^{j}, g_{p, q}^{j}, h_{p, q}^{j}\right\}$ to both $P_j^{\ell-1}$ and $P_j^{\ell}$; for $\ell=1$, the corresponding set $\left\{v_{1, 1}^{j}, u_{1, 1}^{j}, w_{1, 1}^{j}, g_{1, 1}^{j}, h_{1, 1}^{j}\right\}$ lies in only one bag $P_j^1$.
    
    Observe that for each $v\in \CH^j$, there exists a bag of the constructed path $P_j$ that contains $v$; in addition to that, each vertex occupies either all or $2$ consecutive bags. The last property of a decomposition that is needed to be a valid path decomposition is, for each edge of the graph, the existence of a bag with both its endpoints. For edges that are incident to any of $\{x^j, y^j\}\cup\{V_i^j, V_i^{j+1}\}_{i\in [k]}$ the property holds, since these vertices are in all bags, thus, they lie with each other vertex in the same bag at least once. Then, there are edges of the horizontal paths of $\VS$ gadgets. But all edges of each path were added simultaneously to the same bags. Last step is to check vertical edges of the capacity $1$. Since these edges are between vertices of a consecutive horizontal paths, there is a bag where we can find both their vertices. Thus, we have a valid path-decomposition for each $\CH^j$, $j\in [m]$. Notice,  that each bag of it contains no more than $2k+12$ vertices.

    To construct a valid path-decomposition for the whole graph $\calG$, let us consider a path that is a concatenation of $\{P_j\}_{j\in [m]}$ in increasing order. 
    There are no edges in-between vertices of different $\CH$ gadgets and $V$-vertices are the only intersection of vertex sets of $P_j$ and $P_{j+1}$ for each $j\in[m-1]$. Since $\{V_i^j, V_i^{j+1}\}_{i\in [k]}$ are in all bags of $P_j$, then each V-vertex is in at most two paths with sequential subscripts. As the last step, to all $mkN$ bags of a current decomposition we add vertices $s,t$. It follows that there exists a path-decomposition of $\calG$, such that $\pw(\calG)=2k+13$.
\end{proof}

To complete the line of argumentation for \AoNF, we establish that the constructed graph~$\calG$ is planar. We also state an additional property that will be useful later, to ensure that we can enrich the graph without violating planarity.

\begin{lemma} \label{lem:aonf:graph:is:planar}
    The graph~$\calG$ has a planar embedding containing two distinct faces~$f_1, f_2$, such that:
    \begin{itemize}
        \item all vertices of~$\{x^j \mid j \in [m+1]\}$, the source~$s$, and sink~$t$, are incident to~$f_1$, and
        \item all vertices of~$\{y^j \mid j \in [m+1]\}$, the source~$s$, and sink~$t$, are incident to~$f_2$.
    \end{itemize}
\end{lemma}
\begin{proof}
    The overall graph $\calG$ looks like the drawing of \cref{fig:multicolored-k-clique-reduction}; with minor differences which are crucial for planarization; see \cref{fig:multicolored-k-clique-reduction-planar}. The first difference is that instead of $N$ arcs between any pair $V_i^j$ and $V_i^{j+1}$, $i\in [k]$ and $j\in [m]$, we have a $\VS_i^j$ gadget (see~\cref{fig:VS_gadget}). The second is that each column of VS gadgets with the same superscript form the CH gadget, as was shown in~\cref{fig:CH_gadget}. Thus, the statement of the Lemma we prove becomes obvious. Indeed, as we see in~\cref{fig:multicolored-k-clique-reduction}, there exists an embedding for a non-planar instance, such that on one of the faces there are all vertices of~$\{x^j \mid j \in [m+1]\}\cup\{s, t\}$ and on one of the others there are all vertices of~$\{y^j \mid j \in [m+1]\}\cup\{s, t\}.$ But these two steps we described to obtain a planar instance $\calG$ do not significantly change those faces that we care about, see in~\cref{fig:multicolored-k-clique-reduction-planar}. Thus, the embedding with the necessary property still exist for the planar instance $\calG$ as well.
    %So, there exists a planar embedding of $\calG$ such that on one of the faces there are all vertices of~$\{x^j \mid j \in [m+1]\}\cup\{s, t\}$ and on one of the others there are all vertices of~$\{y^j \mid j \in [m+1]\}\cup\{s, t\}.$
    %TODO. \bmp{We need to argue this property holds to be able to use the reduction to produce planar graphs. The additional fact about two faces containing the specified vertices are sufficient to ensure we can insert the super-source and super-sink later without violating planarity. Liana, could you give an argument to prove this lemma? To make this argument efficiently, it may be useful to just make a drawing of what the overall graph looks like and refer to it.}
\end{proof}

\begin{figure}[t]
    \centering
    \includegraphics[page=2,width=\textwidth]{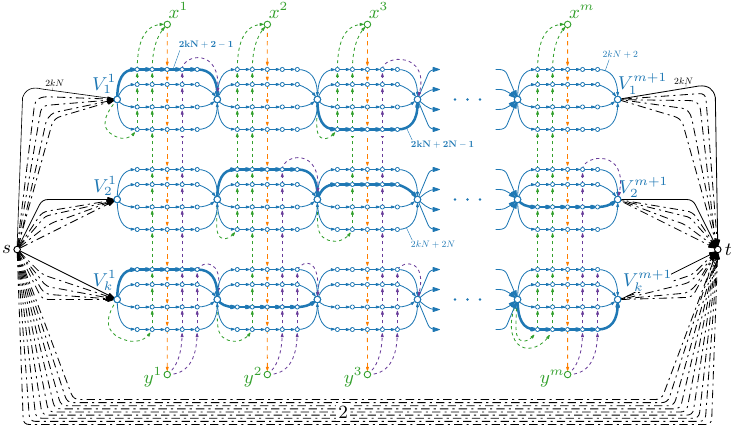}
    \caption{Illustration for the FPT-reduction from \mclique with~$k=3$, $N=4$ to planar \AoNF on a graph of pathwidth~$\Oh(k)$. Dashed edges have capacity~1, dash-dotted edges ($k(N-1)$ of them on the bottom) have capacity~2.}
    \label{fig:multicolored-k-clique-reduction-planar}
\end{figure}

Since the preceding lemmata guarantee that the constructed instance of \AoNF is equivalent to the \mclique instance we started from, its pathwidth is~$\Oh(k)$, and it is planar, this leads to a proof of \cref{lemma:aonf}. Observe that the edge capacities created during the construction are all polynomial in~$k$ and~$N$, both of which are bounded by~$|V(G)|$.

\LemAonf*\label{lemma:aonf*}

\section{From \textsc{All-or-Nothing-Flow} to \textsc{Circulating Orientation}}
To facilitate the next step in our chain of reductions, we make the following observation about the instances constructed by the reduction of \cref{lemma:aonf}. Recall that  \cref{lemma:aonf} transforms an instance~$(G,V_1, \ldots, V_k, k)$ of \mclique into an instance~$(\calG, \calF)$ of \AoNF with source~$s$, sink~$t$, and capacity function~$c$.

\begin{observation} \label{obs:aonf:properties}
    Partition the vertices~$v \in V(\calG) \setminus \{s,t\}$ into the following three sets, based on the difference between the overall capacity of outgoing versus incoming arcs.
        \begin{enumerate}
            \item $V_< := \{v \in V(\calG) \setminus \{s,t\} \mid d^+_\calG(v) < d^-_\calG(v)\}$.
            \item $V_= := \{v \in V(\calG) \setminus \{s,t\} \mid d^+_\calG(v) = d^-_\calG(v)\}$.
            \item $V_> := \{v \in V(\calG) \setminus \{s,t\} \mid d^+_\calG(v) > d^-_\calG(v)\}$.
        \end{enumerate}
        Then the following holds:
        \begin{enumerate}
            \item $V_< = \{x^j \mid j \in [m+1]\} \cup \{V^{m+1}_i \mid i \in [k]\}$.
            \item $V_> = \{y^j \mid i \in [m+1]\} \cup \{V^1_i \mid i \in [k]\}$.
        \end{enumerate}
        Also, $E^-_\calG(s) = E^+_\calG(t) = \emptyset$ and~$d^+_\calG(s) = d^-_\calG(t) = k(2kN+2N) + 2k(N-1)$.
\end{observation}

We use \cref{obs:aonf:properties} in the proof of the next lemma.

\begin{lemma} \label{lemma:co-intermediate}
There is a polynomial-time algorithm that, given an instance of \mclique with parameter~$k$, outputs an equivalent instance of \COr on a planar graph of pathwidth~$\Oh(k)$ whose edge-capacities are bounded by a polynomial in~$|V(G)|$.
\end{lemma}
\begin{proof}
Consider an instance~$(G,V_1, \ldots, V_k, k)$ of \mclique, where each set~$V_i$ has size~$N$. We assume without loss of generality that~$N \geq 10k$, which can be achieved by inserting isolated vertices if needed. By applying the construction of \cref{lemma:aonf}, we transform the instance into an equivalent instance~$(\calG, \calF)$ of \AoNF with source~$s$, sink~$t$, and capacity function~$c$. To transform the latter into an equivalent instance~$(H,c_H)$ of \COr, we proceed as follows.
\begin{itemize}
    \item Initialize~$H$ as the edge-capacitated undirected multigraph obtained from~$\calG$ by simply forgetting the orientation of all edges, whose capacity function we denote by~$c_H$. 
    
    All edges of~$H$ originating from this step are called \emph{standard edges}. The remaining edges of~$H$, to be introduced below, are called \emph{special edges}.
    \item For each~$i \in [k]$, add an edge~$\{V^1_i, s\}$ of capacity~$\alpha := \sum _{q \in [N-1]} 2kN + 2q$. Note that~$d^-_\calG(V^1_i) + \alpha = d^+_\calG(V^1_i)$.
    \item For each~$i \in [k]$, add an edge~$\{t, V^{m+1}_i\}$ of capacity~$\alpha$. Note that~$d^-_\calG(V^{m+1}_i) + \alpha = d^+_\calG(V^{m+1}_i)$.
    \item Add a super-source~$S$, along with an edge~$\{y^j, S\}$ of capacity~$1$ for each~$j \in [m+1]$.
    \item Add a super-sink~$T$, along with an edge~$\{T, x^j\}$ of capacity~$1$ for each~$j \in [m+1]$.
    \item Let~$\beta := 2\calF - d^+_\calG(s)$. Note that this value is non-negative by the construction of~$\calG$; see \cref{obs:aonf:properties}. This guarantees that~$d^+_{\calG}(s) - \calF + \beta = \calF$, which we will use below.
    \item Let~$\xi$ be a sufficiently large value, such as~$\xi := 100 k^2 N^2$. Add the following edges, forming a 4-cycle in the graph~$H$:
    \begin{itemize}
        % \item An edge~$\{S,s\}$ of capacity~$\xi + \beta + k \alpha$.
        % \item An edge~$\{s, T\}$ of capacity~$\xi$.
        % \item An edge~$\{T,t\}$ of capacity~$\xi + m$.
        % \item An edge~$\{t,S\}$ of capacity~$\xi + \beta + k \alpha + m$.
        \item An edge~$\{S,s\}$ of capacity~$\xi + \beta + m$.
        \item An edge~$\{s, T\}$ of capacity~$\xi + k \cdot \alpha + m$.
        \item An edge~$\{T,t\}$ of capacity~$\xi + k \cdot \alpha$.
        \item An edge~$\{t,S\}$ of capacity~$\xi + \beta$.
    \end{itemize}
\end{itemize}
This concludes the description of~$(H,c_H)$. It is not difficult to see that~$H$ is planar: we obtain it as a copy of a planar graph in Step 1 and then add edges which are parallel to existing edges in the following two steps. The remaining steps consist of inserting two vertices~$S$ and~$T$ and edges incident on them. By \cref{lem:aonf:graph:is:planar}, there are two distinct faces in the embedding of~$H$ into which we can draw~$S$ and~$T$ that contain all their prospective neighbors, so that the edges can be drawn without crossings.

It is also not difficult to see that the pathwidth of~$H$ is still~$\Oh(k)$: all edges we create are incident to one of the vertices~$\{s,t,S,T\}$, so by simply inserting these four vertices into all bags of a path decomposition of~$\calG$, we obtain a path decomposition of~$H$ whose pathwidth is larger than the pathwidth~$\Oh(k)$ of~$\calG$ (\cref{lemma:aonf:pathwidth}) by at most four. It remains to prove that the resulting instance of \COr is equivalent to the \AoNF-instance on~$\calG$.

\begin{numclaim} \label{claim:aonf:to:co}
If~$\calG$ has an all-or-nothing flow~$f$ of value~$\calF$, then~$(H, c_H)$ has a circulating orientation.
\end{numclaim}
\begin{proof}
Suppose that there is an all-or-nothing flow~$f$ of value~$\calF$ in the flow network~$\calG$. We use it to obtain a circulating orientation of~$(H,c_H)$, as follows:
\begin{itemize}
    \item For the edges~$e \in E(H)$ which were copied from~$\calG$ in the first step, we orient them as follows: if~$f(e) = c(e)$ then we orient~$e$ in the same way as it appears in~$\calG$; otherwise ($f(e) = 0$) we reverse its orientation.
    \item We orient the edges~$\{V^1_i, s\}$ for~$i \in [k]$ from~$V^1_i$ to~$s$.
    \item We orient the edges~$\{t, V^{m+1}_i\}$ for~$i \in [k]$ from~$t$ to~$V^1_i$.
    \item We orient the edges~$\{y^j, S\}$ for~$j \in [m+1]$ from~$y^j$ to~$S$.
    \item We orient the edges~$\{T, x^j\}$ for~$j \in [m+1]$ from~$T$ to~$x^j$.
    \item We orient the edges of the final 4-cycle as~$(S,s), (s,T), (T,t), (t,S)$.
\end{itemize}
Let~$\vv{H}$ denote the resulting directed graph~$H$. To argue that we have obtained a circulating orientation, we introduce some additional notation. For a vertex~$v \in V(H) \setminus \{S,T\}$, which also exists in~$\calG$, we use~$\hat{d}^+_{\vv{H}}(v)$ ($\hat{d}^-_{\vv{H}}(v)$) to denote the total capacity of the standard edges incident on~$v$ which are oriented out of~$v$ (into~$v$) in~$\vv{H}$. Hence the special edges do not contribute to these terms.

We argue that for~$v \in V(H) \setminus \{S,T, s, t\}$ we have~$\hat{d}^+_{\vv{H}}(v) = d^-_\calG(v)$ and $\hat{d}^-_{\vv{H}}(v) = d^+_\calG(v)$, that is, the contribution of standard edges is independent of the flow~$f$ and is based on the capacity of the edges oriented oppositely in~$\calG$. To see this, note that if~$f$ does not send any flow through~$v$ then all arcs of~$E^+_\calG(v)$ are reversed in~$\vv{H}$, and similarly all arcs of~$E^-_\calG(v)$ are reversed in~$\vv{H}$. For each unit of flow sent into~$v$, there is a unit of capacity on an incoming arc that retains its orientation in~$\vv{H}$, and by flow conservation, a unit of capacity on an outgoing arc that retains its orientation in~$\vv{H}$. Since each edge that carries flow is utilized to full capacity, the choices made by~$f$ do not influence the contribution of standard edges to inwards and outwards oriented capacity of~$v$. 

Based on this insight, we now prove that~$\vv{H}$ is a circulating orientation:~$d^+_{\vv{H}}(v) = d^-_{\vv{H}}(v)$ for each~$v \in V(H)$. For this argument, we distinguish several types of vertices, making use of the partition of~$V(\calG) \setminus \{s,t\}$ into sets~$V_<, V_=, V_>$ given in \cref{obs:aonf:properties}.
\begin{enumerate}
    \item Consider a vertex~$v \in V(H) \cap V_=$, i.e. a vertex of~$H$ that was already present in the flow network~$\calG$, in which the capacity of its incoming arcs equalled that of its outgoing arcs. In the construction of~$H$ we did not add any edge incident on~$v$, so that~$d^+_{\vv{H}}(v) = \hat{d}^+_{\vv{H}}(v)$ and~$d^-_{\vv{H}}(v) = \hat{d}^-_{\vv{H}}(v)$; all its incident edges are standard edges. By the argument above, we have 
    \[d^-_{\vv{H}}(v) = \hat{d}^-_{\vv{H}}(v) = d^+_{\calG}(v) = d^-_{\calG}(v) = \hat{d}^+_{\vv{H}}(v) = d^+_{\vv{H}}(v),\]
    where the middle equality comes from the definition of the set~$V_=$.
    \item Consider a vertex~$x^j$ for~$j \in [m+1]$. By the argument above we have~$\hat{d}^+_{\vv{H}}(x^j) = d^-_\calG(x^j) = 2$, while $\hat{d}^-_{\vv{H}}(x^j) = d^+_\calG(x^j) = 1$; here we use the construction of the instance~$\calG$. Since~$x^j$ is incident on one special edge of capacity~$1$, which is oriented from~$T$ to~$x^j$, this gives~$d^+_{\vv{H}}(x^j) = d^-_{\vv{H}}(x^j) = 2$. The argument for~$y^j$ is symmetric.
    \item Consider a vertex~$V^1_i$ for~$i \in [k]$. By the argument above we have $\hat{d}^+_{\vv{H}}(V^1_i) = d^-_\calG(V^1_i)$, which equals~$2kN+2N$ by construction of~$\calG$. We have $\hat{d}^-_{\vv{H}}(V^1_i) = d^+_\calG(V^1_i)$, which equals~$\sum_{q \in [N]} 2kN+2q$. The unique special edge~$\{V^1_i, s\}$ incident on~$V^1_i$ has capacity~$\alpha = \sum _{q \in [N-1]} 2kN + 2q$ and is oriented towards~$s$. It therefore compensates for the imbalance and implies~$d^+_{\vv{H}}(V^1_i) = d^-_{\vv{H}}(V^1_i)$. The argument for $V^{m+1}_i$ is symmetric.
    \item Consider the super-source~$S$. All its incident edges are oriented inwards, except for the edge~$\{s,S\}$ which has capacity~$\xi + \beta + m$. Since the total capacity of the remaining edges is~$\xi + \beta$ (for the edge~$\{t, S\}$) plus~$1 \cdot m$ (for the~$m$ edges~$\{y^j, S\}$ of capacity~$1$), vertex~$S$ satisfies the requirements.
    \item The argument for~$T$ is similar: it has a single edge~$(s,T)$ of capacity~$\xi + k\cdot \alpha + m$ oriented inwards, while the total capacity of the remaining outwards-oriented edges is~$\xi + k \cdot \alpha$ (for~$\{T,t\}$) plus~$1 \cdot m$ (for the $m$ edges~$\{T, x^j\}$).
    \item Consider the normal source~$s$ of the flow network~$\calG$. From the standard edges, the fact that the flow has value~$\calF$ while the source has no outgoing arcs in~$\calG$ ensures that there is~$\calF$ capacity on the outwards oriented edges, while the remaining~$d^+_{\calG}(s) - \calF$ capacity is oriented inwards. 
    
    From the special edges incident on~$s$, we have~$\xi + \beta + m$ capacity oriented into~$s$ on the edge~$\{S,s\}$ plus~$k \cdot \alpha$ capacity on the edges~$\{V^1_i, s\}$ for~$i \in [k]$. The only special edge oriented out of~$s$ is~$\{s,T\}$ of capacity~$\xi + k \cdot \alpha + m$. We therefore find:
    \begin{align*}
        d^-_{\vv{H}}(s) &= (d^+_{\calG}(s) - \calF) + (\xi + \beta + m) + (k \cdot \alpha) & \mbox{Standard plus special edges.} \\
        &= \calF + (\xi + k \cdot \alpha + m) & \mbox{As $d^+_{\calG}(s) - \calF + \beta = \calF$.} \\
        &= d^+_{\vv{H}}(s). & \mbox{Standard plus special edges.}
    \end{align*}
    
    \item Finally, consider the normal sink~$t$ of~$\calG$. From the standard edges, there is~$\calF$ capacity on edges into~$t$ and~$d^-_\calG(t) - \calF = d^+_\calG(s) - \calF$ capacity on edges out of~$t$. From the special edges, there is~$\xi + k \cdot \alpha$ capacity on the edge~$(T,t)$ entering~$t$, while the capacity of special edges leaving~$t$ is~$k \cdot \alpha$ (for~$V^{m+1}_i$ for each~$i \in [k]$) plus~$\xi + \beta$ (for~$\{T,t\}$), for a total of~$k \cdot \alpha + \xi + \beta$. We therefore find:
    \begin{align*}
        d^-_{\vv{H}}(t) &= \calF + \xi + k \cdot \alpha & \mbox{Standard plus special edges.} \\
        &= d^+_\calG(s) - \calF + k \cdot \alpha + \xi + \beta & \mbox{As $d^+_{\calG}(s) - \calF + \beta = \calF$.} \\
        &= d^+_{\vv{H}}(t). & \mbox{Standard plus special edges.}
    \end{align*}
\end{enumerate}
\cref{obs:aonf:properties} shows that the treated cases indeed cover all vertices of~$H$. As each type of vertex has the same capacity oriented inwards and outwards, the given orientation~$\vv{H}$ is indeed a circular orientation. This proves the claim.
\end{proof}

We now prove the converse.

\begin{numclaim} \label{claim:co:to:aonf}
If~$(H, c_H)$ has a circulating orientation, then~$\calG$ has an all-or-nothing flow~$f$ of value~$\calF$.
\end{numclaim}
\begin{proof}
Consider a circulating orientation~$\vv{H}$ of~$H$. We argue that for each of the four vertices~$\{s, S, t, T\}$, exactly one edge of the 4-cycle on these vertices is oriented inwards and one edge is oriented outwards. This follows from the fact that all edges on this 4-cycle have capacity at least~$\xi$, and that~$\xi$ is larger than the sum of the capacities of all other edges incident on a common vertex. Hence if the two edges of capacity at least~$\xi$ are both oriented into a vertex, or both oriented out of a vertex, then the capacity of the remaining edges is not sufficient to obtain equal values for the inwards- versus outwards-oriented capacity. 

The preceding argument shows that the 4-cycle is either oriented as~$(S,s), (s,T), (T,t), (t,S)$ or~$(S,s), (s,T), (T,t), (t,S)$ $(S,t), (t, T), (T, s), (s, S)$. In the latter case, we reverse the orientation of all edges (which preserves a circulating orientation). Hence we may assume from now on that in~$\vv{H}$, the 4-cycle is oriented as~$(S,s), (s,T), (T,t), (t,S)$.

Based on this orientation, we define a flow in~$\calG$ as follows. For each arc~$e \in E(\calG)$, if~$\vv{H}$ orients the corresponding edge of~$H$ in the same direction as it appears in~$\calG$, then we define~$f(e) = c(e)$; otherwise we define~$f(e) = 0$. It remains to show that~$f$ has the flow conservation property and defines a flow of value~$\calF$. Before we do so, we establish some consequences of this orientation of the 4-cycle.

\begin{enumerate}
    \item Each special edge~$\{T, x^j\}$ for~$i \in [k]$ is oriented away from~$T$ in~$\vv{H}$. To see this, observe that the capacity~$\xi + k \cdot \alpha + m$ of the edge~$(s,T)$ on the 4-cycle which is oriented into~$T$, equals the capacity of all other edges incident on~$T$ combined. Hence, to achieve a circulating orientation the remaining edges incident to~$T$ (including all edges~$\{T, x^j\}$) are oriented out of~$T$.\label{co:xj}
    \item Each special edge~$\{y^j, S\}$ for~$i \in [k]$ is oriented into~$S$ in~$\vv{H}$. The argument is similar as in the previous case: the capacity~$\xi + \beta + m$ on the edge~$(S,s)$ of the 4-cycle which is oriented out of~$S$ by~$\vv{H}$ is equal to the total capacity of the remaining edges incident to~$S$ (including~$\{y^j, S\})$, which are therefore oriented into~$S$.\label{co:yj}
    \item Each edge~$\{V^1_i, s\}$ for~$i \in [k]$ is oriented into~$s$. Here the argument is slightly more delicate, but the idea is similar. Among the edges of the 4-cycle incident to~$s$, the edge~$\{S,s\}$ of capacity~$\xi + \beta + m$ is oriented into~$s$ while the edge~$\{s,T\}$ of capacity~$\xi + k\cdot \alpha + m$ is oriented out of~$s$. It turns out that the value of~$\alpha$ is so large compared to the capacities of the remaining edges incident to~$s$ that, in order to get a circulating orientation, all~$k$ edges of capacity~$\alpha$ must be oriented into~$s$. Namely, suppose for a contradiction that at least one of the~$k$ edges of capacity~$\alpha$ is oriented away from~$s$. Then the capacity of edges oriented away from~$s$ is at least~$(\xi + k \cdot \alpha + m) + \alpha$, while the capacity of the remaining edges incident on~$s$ is at most~$\xi + \beta + m + k(2kN + 2N) + (k-1)\alpha + k(2(N-1))$, which can be verified to be strictly smaller since~$\alpha$ is quadratic in~$N$. (Recall the assumption that $N \geq 10k$.)\label{co:vone}
    \item A symmetric argument, again based on the fact that~$\alpha$ is quadratic in~$N$, ensures that each edge~$\{t, V^{m+1}_i\}$ is oriented out of~$t$.\label{co:vm}
\end{enumerate}

% TODO: Insert an argument saying that if d^+_{\calG}(v) = d^-_{\calG}(v), then $f$ satisfies flow conservation for $v$. We could say it is a special case, treating it as a special edge of capacity 0.

\paragraph{A sufficient condition for flow conservation.} We make the following general claim, which will help us establish flow conservation. If~$v \in V(\calG) \setminus \{s,t\}$ satisfies one of the following:
\begin{enumerate}
    \item $v$ is incident on a unique special edge~$e^*$, that special edge is oriented out of~$v$ in~$\vv{H}$, and~$c_H(e^*) + d^-_{\calG}(v) = d^+_{\calG}(v)$,\label[condition]{co:special:out}
    \item $v$ is incident on a unique special edge~$e^*$, that special edge is oriented into~$v$ in~$\vv{H}$ and~$c_H(e^*) + d^+_{\calG}(v) = d^-_{\calG}(v)$, or\label[condition]{co:special:in}
    \item $v$ is not incident on any special edge and~$d^+_{\calG}(v) = d^-{\calG}(v)$,\label[condition]{co:nospecial}
\end{enumerate}
then~$f$ satisfies the flow conservation property for vertex~$v$. 

We prove this implication as follows. Suppose the first condition holds for~$v$. Recall that the total capacity of edges incident on~$v$ is~$d_H(v)$, which equals~$c_H(e^*) + d^-_{\calG}(v) + d^+_{\calG}(v)$ since~$e^*$ is the unique special edge incident on~$v$. By the assumption above, we have~$d_H(v)/2 = d^+_{\calG}(v)$. Consider the subset of~$E^-_{\calG}(v)$ whose corresponding edges are oriented into~$v$ in~$\vv{H}$ and let their total capacity be~$D$. This means that~$f$ sends flow over these edges into~$v$ and therefore the flow into~$v$ under~$f$ equals~$D$. Since~$\vv{H}$ is a circular orientation, half of the capacity incident on~$v$ is oriented into~$v$. Since no special edge is oriented into~$v$, while the edges corresponding to~$E^-_{\calG}(v)$ contribute~$D$ to the capacity oriented into~$v$, from the remaining edges (corresponding to~$E^+_{\calG}(v)$) there is a total of~$d_H(v)/2 - D = d^+_{\calG}(v) - D$ capacity oriented into~$v$. It follows that the rest of the edges of~$E^+_{\calG}(v)$ are oriented out of~$v$, which means~$d^+_{\calG}(v) - (d^+_{\calG}(v) - D) = D$ capacity of the edges of~$E^+_{\calG}(v)$ are oriented out of~$v$ in~$\vv{H}$ and therefore have the same orientation in~$\vv{H}$ as in~$\calG$. Consequently,~$f$ sends~$D$ flow out of~$v$ over these edges, which implies that flow conservation holds for~$v$. This argument establishes the first condition. The proof of the second condition is symmetric. The proof for the third condition follows by the same reasoning, for example by considering a single hypothetical capacity-0 special edge incident on~$v$ which does not affect any argument.

\paragraph{Flow conservation.} Using the implication derived above to establish that~$f$ has the flow conservation property for all~$v \in V(\calG) \setminus \{s,t\}$, it suffices to establish that one of the three cases above holds for~$v$. We distinguish several types of vertices of~$V(\calG) \setminus \{s,t\}$, based on the partition given in \cref{obs:aonf:properties}.

\begin{itemize}
    \item For~$v \in V_=$, by definition we have~$d^+_{\calG}(v) = d^-_{\calG}(v)$, so that \cref{co:nospecial} applies to guarantee flow conservation for~$v$.
    \item For each vertex~$x^j$ for~$j \in [m]$, by construction of~$H$ there is a unique special edge~$\{T, x^j\}$ incident on~$x^j$ which has capacity~$1$ and is oriented into~$x^j$. As~$d^+_{\calG}(x^j) = 1$ and~$x^-_{\calG}(x^j) = 2$, \cref{co:special:in} guarantees flow conservation for~$x^j$. The argument for~$y^j$ is symmetric.
    \item For each vertex~$V^1_i$ for~$i \in [k]$, by construction of~$H$ there is a unique special edge~$\{V^1_i, s\}$ incident on~$V^1_i$ which has capacity~$\alpha$ and is oriented out of~$V^1_i$. As~$d^+_{\calG}(V^1_i) = d^-_{\calG}(v^1_i) + \alpha$ by choice of~$\alpha$, \cref{co:special:out} guarantees flow conservation for~$V^1_i$. The argument for~$V^{m+1}_i$ is symmetric.
\end{itemize}
As \cref{obs:aonf:properties} shows that this covers all vertices of~$V(\calG) \setminus \{s,t\}$, we conclude that~$f$ has the flow conservation property.

\paragraph{Value of the flow.} As the final step of the argument, we prove that the value of flow~$f$ is~$\calF$, that is, the capacity of edges leaving~$s$ over which~$f$ sends (the full capacity) of flow equals~$\calF$ and the remainder of the edges out of~$s$ carries no flow. To do so, we analyze the orientation of the edges incident on~$s$ in~$\vv{H}$. Out of the special edges incident on~$s$, we know that:
\begin{itemize}
    \item a capacity of $\xi + \beta + m + k \cdot \alpha$ is oriented into $s$ (by the edge $\{s,S\}$ and the~$k$ edges~$\{V^1_i, s\}$), and
    \item a capacity of $\xi + k \cdot \alpha + m$ is oriented out of $s$ (by the edge $\{s,T\}$).
\end{itemize}
Since~$\vv{H}$ is a circulating orientation, the standard edges therefore contribute $\beta$ more out-capacity than in-capacity of~$s$ in~$\vv{H}$. The total capacity of the standard edges incident on $s$ is $d^+_{\calG}(s)$. Let $D$ be the capacity of the standard edges oriented outwards, implying that the remaining $d^+_{\calG}(s) - D$ of standard capacity is oriented inwards. Note that the flow out of~$s$ under~$f$ is exactly~$D$. Since the standard out-capacity exceeds the standard in-capacity by~$\beta$, we have that~$D = \beta + (d^+_{\calG}(s) - D)$, which means that~$2D = \beta + d^+_{\calG}(s)$. Since $\beta = 2\calF - d^+_{\calG}(s)$, this implies~$2D = 2\calF - d^+_{\calG}(s) + d^+_{\calG}(s)$, so that~$D = \calF$ and the flow out of the source~$s$ indeed equals~$\calF$ in~$f$.
\end{proof}

Since the preceding two claims show that the constructed instance is equivalent to the input instance, while we already argued for its planarity and bounded pathwidth earlier, this completes the proof of \cref{lemma:co-intermediate}.
\end{proof}

The proof of the following lemma now easily follows by some post-processing steps that can be done in a black-box fashion. To obtain an instance with the stated properties, we have to turn the planar multigraph resulting from \cref{lemma:co-intermediate} into a simple graph, make it triconnected, and triangulate it.

\LemCo*\label{lemma:co*}
\begin{proof}
Consider an instance~$(G,V_1, \ldots, V_k, k)$ of \mclique. By the construction of \cref{lemma:co-intermediate} we transform such instance into an equivalent instance~$(H, c_H)$ of \COr on a planar multigraph of pathwidth~$\Oh(k)$. 
We then subdivide all edges of~$H$: For each edge~$e = \{u,v\} \in E(H)$ we introduce a new vertex~$w_e$ and replace~$e$ by a pair of edges~$\{u, w_e\}, \{w_e, v\}$, both of capacity~$c_H(e)$. It is easy to see that this subdivision does not affect the existence of a circulating orientation, and it is well-known that subdividing all edges increases the pathwidth by at most one; this also follows from \cref{lemma:pw_gadgets}. Let~$(H', c_{H'})$ be the resulting edge-capacitated simple planar graph. 

We apply \cref{theorem:pw_triangulation} to the simple graph~$H'$ to obtain a simple plane triangulation~$H''$ such that~$\pw(H'') \in \Oh(k)$. We define a capacity function~$c_{H''}$ for~$E(H'')$ as follows: we set~$c_{H''}(e) = c_{H''}(e)$ for all~$e \in E(H') \cap E(H'')$, while we set~$c_{H''}(e) = 0$ for~$e \notin E(H')$. Hence, all newly created edges obtain a capacity of 0 so that their presence does not affect whether an orientation of the remaining edges is circulating or not. The instance~$(H'', c_{H''})$ therefore satisfies all requirements and is equivalent to the original \mclique instance. As all steps of the process can be done in polynomial time, the lemma follows.
\end{proof}

\section{From \textsc{Circulating Orientation} to \textsc{Upward Planarity Testing}}

Here we provide the formal proof of the reduction from \textsc{Circulating Orientation} to \textsc{Upward Planarity Testing}, encapsulated in the following lemma.

%\todo[inline]{Optionally add the next step of the figure: correspondence between orientation of edges incident to a vertex in $P$, orientation of edges on the boundary walk on the dual face in $D$, and flip of the respective tendrils.}

\begin{lemma} \label{lemma:upwards}
There is a polynomial-time reduction from planar triangulated triconnected instances of \textsc{Circulating Orientation} of pathwidth~$k$ with capacities polynomial in graph size to  \textsc{Upward Planarity Testing} instances of pathwidth~$\Oh(k)$.
\end{lemma}
\begin{proof}
    Let $(P, c)$ be an instance of \textsc{Circulating Orientation}, where $P$ is a triconnected planar graph, and $c: E(P) \to \mathbb{Z}_{\ge 0}$ is the capacity function.
    %By Theorem~\ref{theorem:pw_triangulation}, w.l.o.g. we can assume that $P$ is triangulated, as adding zero-weight edges to an instance of \textsc{Circulating Orientation} is safe.\bmpr{You do not have to invoke Biedl again here; you can start from \cref{lemma:co*} which guarantees a triconnected planar triangulated graph.}
    W.l.o.g., we can also assume that every capacity is a multiple of $3$, as multiplying all capacities by the same factor is safe.

    Let $\Epsilon_P$ be a planar embedding of $P$, and let $D$ be the dual graph of $P$ with respect to $\Epsilon_P$; see also \cref{fig:upward-P,fig:upward-D} for an illustration. Let $\Epsilon_D$ be the respective planar embedding of $D$, and let $w : E(D) \to \mathbb{Z}_{\ge 0}$ be the weight function on the edges of $D$ defined so that $w(e') = c(e)$, where $e \in E(P)$, $e' \in E(D)$, and $e'$ is the dual edge of $e$.
    \begin{numclaim}
        The graph $D$ is a planar triconnected graph of maximum degree 3 and pathwidth $\Oh(k)$.
    \end{numclaim}
    \begin{proof}
        Since $D$ is the dual of the planar triconnected graph $P$, it is immediately planar and also triconnected. Since $P$ is triangulated, the maximum degree of its dual is 3. Pathwidth remains bounded by~\cref{theorem:pw_dual}. %Planar holds by definition, triconnected by Steinitz's theorem, and %Theorem~\ref{thm:embeddingtw}.
    \end{proof}

    First, we pick an orientation of $D$ so that the resulting graph is $st$-planar, see, for instance, ~\cref{fig:upward-VD}. This will be helpful to restrict possible embeddings of the target graph of the reduction.

    \begin{numclaim}\label{claim:D_unique_emb}
        There exists an orientation $\vv{D}$ of $D$ such that $\vv{D}$ is an $st$-planar graph that admits a unique upward planar embedding.
        Also, its underlying planar embedding is the unique planar embedding of $D$ where $s$ and $t$ lie on the outer face.
    \end{numclaim}
    \begin{proof}
        Pick arbitrarily an adjacent pair of vertices $s \ne t \in V(D)$. Since $D$ is triconnected, there exists an $st$-numbering of $V(D)$~\cite{LempelEC67}, and so an orientation $\vv{D}$ of $D$ such that $\vv{D}$ is acyclic, $s$ is the only source, and $t$ is the only sink. Moreover, since $D$ is triconnected, its planar embedding is fixed up to the choice of the outer face. Since $s$ is the only source and $t$ is the only sink, they both have to be part of the boundary of the outer face.
        By triconnectivity of $D$, there is only one face where both $s$ and $t$ lie on the boundary.
    \end{proof}

    We now construct the target digraph $\vv{G}$ of our reduction. Consider the digraph $\vv{D}$, and replace every arc $uv \in E(\vv{D})$ with a copy of the tendril $T_k$, where $k = w(uv)$, such that $u$ is the source pole of $T_k$, and $v$ is the sink pole. The resulting digraph, after this replacement is performed for all the arcs of $\vv{D}$, is precisely $\vv{G}$. We call the vertices of $\vv{G}$ that originate from vertices of $\vv{D}$ \emph{base} vertices, and denote the respective set by $B \subset V(G)$. The other vertices, i.e., the inner vertices of the tendrils, are called \emph{auxiliary}, denoted by $A \subset V(\vv{G})$. For base vertices $u, v \in B$, if $uv \in E(\vv{D})$, we denote the set of vertices of the respective tendril in $G$ by $T^{uv}$. Here and next for the sake of readability we associate base vertices of $\vv{G}$ with the respective vertices of $\vv{D}$, slightly abusing the notation. By definition, it holds that $A = \bigcup_{uv \in E(\vv{D})} T^{uv} \setminus \{u, v\}$.
    
    By construction, we observe several immediate properties of the digraph $\vv{G}$, see also~\cref{lemma:tendrils}.
    
    \begin{numclaim}
        The digraph $\vv{G}$ is acyclic. Its size is polynomial in the size of the digraph $\vv{D}$ and weights $w$, and its pathwidth is $\Oh(k)$.
    \end{numclaim}

    We now characterize the planar embeddings of $\vv{G}$ where $s$ and $t$ belong to the outer face. Let $\vv{D}'$ be an arbitrary orientation of $D$ (not necessarily $\vv{D}$). We say that $\vv{D}'$ defines an embedding $\Epsilon_{\vv{D}'}$ of $\vv{G}$ in the following sense. The embedding of $\vv{G}$ is obtained by enhancing the embedding of $\vv{D}$ given by~\cref{claim:D_unique_emb}: every arc $e \in E(\vv{D})$ is replaced with an embedding of the respective tendril $T^{uv}$. We pick a flip of this embedding consistently with the orientation $\vv{D}'$: for an arc $e \in E(\vv{D}')$, let the left face and the right face be defined with respect to the direction of $e$, in the unique upward planar embedding of $\vv{D}$. We pick the embedding of $T_{c(e)}$ so that the boundary with positive contribution goes on the right face, while the negative contribution to the left face. In other words, for a face $f$ that has an arc $e$ on its boundary walk, the contribution of the respective tendril will be positive if $e$ is oriented clockwise along the walk, and negative otherwise; see \cref{fig:upward-T,fig:upward-G} for an illustration.
    
    Later, in this way we will model the choice of orienting the edge in $P$ by the choice of orienting its dual in $D$, and show that the orientation achieves the conservation of flow in every vertex if and only if the tendrils flipped according to the respective $\vv{D}'$ provide the total contribution of zero to every face.
    Before we move to the proof of that equivalence, we show that the choice of $\vv{D}'$ essentially fixes the embedding of $\vv{G}$.
    %if $e$ is oriented in the same direction by both $\vv{D}$ and $\vv{D}'$, we pick the embedding of $T_{w_D(e)}$ so that the boundary with positive contribution goes on the left; otherwise, if $e$ is oriented in the other direction by $\vv{D}'$, the boundary with positive contribution goes on the right.
    
    \begin{numclaim}\label{claim:embeddings_by_flips}
        The planar embeddings of $\vv{G}$ are precisely the embeddings $\Epsilon_{\vv{D}'}$ for all possible choices of the orientation $\vv{D}'$ of $D$, up to a change of the outer face. The planar embeddings of $\vv{G}$ where $s$ and $t$ are on the outer face are precisely $\Epsilon_{\vv{D}'}$.
    \end{numclaim}
    \begin{proof}
        The claim follows immediately from the fact that $\vv{G}$ is obtained by replacing edges of the triconnected graph $\vv{D}$ by triconnected\footnote{accounting also for connectivity through the rest of the graph} tendrils $T_k$, thus the planar embedding of $\vv{G}$ is predetermined up to a flip of each tendril. The second part of the statement  follows since setting $s$ and $t$ to the boundary of the outer face defines it uniquely (because $\vv{G}$ is triconnected).
    \end{proof}

    For a planar embedding $\Epsilon_{\vv{D}'}$ of $G$, we call a face a \emph{base} face when it originates from a face of $D$, and an \emph{auxiliary} face when it is an inner face of a tendril $\vv{G}[T^{uv}]$ for some $uv \in E(D)$. We call all angles that belong to auxiliary faces or auxiliary vertices \emph{auxiliary} angles, and the remaining angles, which are the angles of base vertices in base faces, \emph{base} angles. Again, with a slight abuse of notation, we associate base faces and base angles in $G$ with the respective faces and angles in $D$. Next we argue that given a fixed flip of the tendrils, the angle assignment on $\vv{G}$ is completely predetermined by unique upward planar embeddings of the tendrils and of $\vv{D}$.

    \begin{numclaim}\label{claim:lambda_fixed}
        Let $\vv{G}$ admit an upward planar embedding with the planar embedding $\Epsilon$ and the angle assignment $\lambda$. Then $s$ and $t$ lie on the outer face in this embedding. Moreover, on every auxiliary angle, $\lambda$ coincides with the unique upward planar embedding of the respective tendril. On every base angle, $\lambda$ coincides with the unique upward planar embedding of $\vv{D}$.
    \end{numclaim}
    \begin{proof}
        For the angle assignment, let $\vv{T} = \vv{G}[T^{uv}]$ be a tendril in $\vv{G}$ for some $uv \in E(D)$. The upward planar embedding of $\vv{G}$ remains upward planar when restricted to $\vv{T}$, thus the respective angle assignment has to coincide with the unique upward planar embedding of $\vv{T}$. This proves the statement for all auxiliary angles.

        Let $v$ be a base vertex of $\vv{G}$, $v \ne s, t$. Then it is a switch vertex in both $\vv{D}$ and $\vv{G}$, and the positions for flat angles
        are predetermined by changes of edge directions in the circular order given by the embedding $\Epsilon$. These positions match exactly the flat angles of $v$ in the embedding of $\vv{D}$, as each tendril appears in the order consecutively by~\cref{claim:embeddings_by_flips}, and the edges of a tendril incident to $v$ all have the same direction as the original edge in $\vv{D}$.

        We now show that $s$ and $t$ lie on the outer face. Assume the contrary, then the outer face is either an auxiliary face or a base face that is inner in the embedding of $\vv{D}$. The first case is immediately impossible, as an auxiliary face is an inner face of a tendril, and the tendril admits a unique upward planar embedding. In the second case, we consider the contributions of angles to the face in the upward planar embedding of $\vv{G}$. For base angles not from $s$ or $t$, by the above the angle assignment coincides with the unique upward planar embedding of $\vv{D}$, thus every base angle is flat except for the local source and local sink of the face. Since at most one of them is in $\{s, t\}$, at most one of them is large, and so the total contribution of base angles is either $0$ or $-2$. For the auxiliary angles, their values are, as above, predetermined by the unique upward planar embedding of each tendril. It remains to observe that the contribution to the face of auxiliary angles of a tendril $T_{w(e)}$ on the boundary of the face is either $-6k$ or $6k$ for some $k \in \mathbb{Z}_{\ge 0}$ by~\cref{lemma:tendrils} and the assumption that every capacity in the input graph $P$ is a multiple of $3$. On the other hand, by the property \textbf{UP3} of an upward planar embedding the total contribution to the outer face must be 2. Since the base angles contribute 0 or -2 in total, and the auxiliary angles contribute the total of $6k$, for some $k \in \mathbb{Z}$,  the value of $2$ cannot be reached independently of which tendrils lie on the boundary, and their flips. We arrive at a contradiction, therefore both $s$ and $t$ must lie on the outer face of the embedding.
        
        Finally, it remains to show that the angle assignment for $s$ and $t$ coincides with the upward planar embedding of $\vv{D}$, i.e, that the large angle for both $s$ and $t$ lies on the outer face.
        Assume $v = s$ and the large angle is assigned to an inner base face $f$ instead of the outer face. By the above, all the remaining base angles on $f$ are defined by the upward planar embedding of $\vv{D}$, and thus exactly one of them is small and the rest are flat. Thus, the balance of the base angles on $f$ is zero. Again, since the contribution of auxiliary angles to $f$ is a multiple of $6$, the target value of $-2$ required by the property \textbf{UP3} of an upward planar embedding cannot be reached, which is a contradiction. The case $v = t$ is symmetric.
    \end{proof}
    
    With that, we are ready to show that $\vv{G}$ admits an upward planar embedding if and only if there exists a circulating orientation of $(P, c)$.

    Fix an upward planar embedding of $\vv{G}$, let $\Epsilon$ be the respective planar embedding of $\vv{G}$, and $\lambda$ be the respective angle assignment.
    By~\cref{claim:lambda_fixed}, $s$ and $t$ belong to the outer face in $\Epsilon$. Thus by~\cref{claim:embeddings_by_flips}, the embedding is given by an orientation $\vv{D}'$ of $D$, that is, $\Epsilon = \Epsilon_{\vv{D}'}$. By~\cref{claim:lambda_fixed}, the angle assignment $\lambda$ is fixed.
    We now show that the following orientation $\vv{P}$ of $P$ constitutes a solution to the instance $(P, c)$ of \textsc{Circulating Orientation}.
    For an edge $e = uv \in E(P)$, orient it away from $u$ if and only if the dual edge $e^* \in E(D)$ is oriented clockwise along the boundary of $f_u$ in $\vv{D}'$, where $f_u$ is the face in $D$ corresponding to the vertex $u \in V(P)$, see~\cref{fig:backtop} for an example. 
    
    %For an edge $uv \in E(P)$, consider the dual edge $f_uf_v \in E(D)$. If the orientation of 
    Consider a base face $f$ of $\vv{G}$, by the characterization of $\lambda$ from~\cref{claim:lambda_fixed} the total contribution of auxiliary angles from tendrils to $f$ must be zero as the condition \textbf{UP3} of~\cref{th:upward-conditions} is fulfilled for $f$ in both $\vv{D}$ and $\vv{G}$.
    Consider the vertex $v_f$ in $P$ corresponding to the face $f$. We get $\sum_{e \in E^+_{\vv{P}}(v_f)} c(e) = \sum_{e' \in E^-_{\vv{P}}(v_f)} c(e')$ from the fact that the total contribution of tendrils to the face $f$ is zero, since an outgoing arc $e \in E^+_{\vv{P}}(v_f)$ of capacity $c(e)$ corresponds to a tendril with a contribution of $2c(e)$ to $f$, and  an incoming arc $e' \in E^-_{\vv{P}}(v_f)$ of capacity $c(e')$ corresponds to a tendril with a contribution of $-2c(e')$ to $f$. Since the above holds for every base face $f$ of $\vv{G}$ and so for every vertex $v_f$ of $P$, the forward direction of the proof is done.

    In the other direction, let $\vv{P}$ be the orientation of $P$ such that for every $v \in V(\vv{P})$, $\sum_{e \in E^+_{\vv{P}}(v)} c(e) = \sum_{e' \in E^-_{\vv{P}}(v)} c(e')$.
    We now argue that this implies an orientation $\vv{D}'$ of $D$ with the following property: for every face $f$ of $D$ (recall the planar embedding is fixed by $\Epsilon_D$), the total sum of the weights of the edges oriented clockwise along the boundary walk of $f$ is equal to the total sum of the weights of the edges oriented counter-clockwise along the boundary walk of $f$. Specifically, this orientation is obtained by orienting a dual edge $e \in E(D)$ of the arc $uv \in E(\vv{P})$ clockwise in the order defined by the boundary walk of $f_u$, which is the face in $D$ dual to the vertex $u \in V(\vv{P})$. This also implies that $e$ is oriented counter-clockwise in the order of the boundary walk of $f_v$, the dual face of $v$. Now, for a face $f_v$ in $D$, which corresponds to a vertex $v$ in $P$, the edges oriented clockwise along its boundary walk in $\vv{D}'$ are precisely those the duals of which are oriented away from $v$ in $\vv{P}$, thus their total weight is $\sum_{e \in E^+_{\vv{P}}(v)} c(e)$. Similarly, an edge is oriented counter-clockwise if and only if its dual is oriented towards $v$ in $\vv{P}$, thus the total weight of counter-clockwise edges is $\sum_{e' \in E^-_{\vv{P}}(v)} c(e')$. Finally, we obtain the desired property by the fact that $\vv{P}$ is a solution so $\sum_{e \in E^+_{\vv{P}}(v)} c(e) = \sum_{e' \in E^-_{\vv{P}}(v)} c(e')$ holds.
    
    We now consider an upward planar embedding $\Epsilon_{\vv{D}'}$ given by the orientation $\vv{D}'$ above, from~\cref{claim:embeddings_by_flips} and~\cref{claim:D_unique_emb}. From~\cref{claim:D_unique_emb}, it is enough to verify that the property \textbf{UP3} for every base face is fulfilled.
    For a base face $f$, the contribution of base angles is the same in $\vv{D}$ and $\vv{G}$ by~\cref{claim:D_unique_emb}, thus \textbf{UP3} is fulfilled for $f$ if the total contribution of auxiliary angles to $f$ is zero. Let $v_f$ be the corresponding to $f$ vertex in $P$, by the property above the sum of weights of edges oriented clockwise along the boundary of $f$ is $\sum_{e \in E^+_{\vv{P}}(v_f)} c(e)$, and counter-clockwise is $\sum_{e \in E^-_{\vv{P}}(v_f)} c(e)$.
    By definition of $\Epsilon_{\vv{D}'}$, every arc $e \in E(\vv{D}')$ oriented clockwise corresponds to a tendril that has a total contribution of $2c(e)$ to $f$, and every arc $e' \in E(\vv{D}')$ oriented counter-clockwise contributes $-2c(e')$. Since $\sum_{e \in E^+_{\vv{P}}(v_f)} c(e) = \sum_{e' \in E^-_{\vv{P}}(v_f)} c(e')$, the total contribution of auxiliary angles to $f$ is thus $2 \sum_{e \in E^+_{\vv{P}}(v_f)} c(e) - 2\sum_{e' \in E^-_{\vv{P}}(v_f)} c(e') = 0$. This finishes the proof of correctness.
    %We now construct an upward planar embedding $\Epsilon$ of $\vv{G}$. Intuitively, we enhance the embedding of $\vv{D}$ by replacing every arc $e \in E(\vv{D})$ with an embedding of the respective copy of the tendril $T_{w_D(e)}$. We also pick a flip of this embedding consistently with the orientation $\vv{D}'$ constructed from the solution $\vv{P}$: if $e$ is oriented in the same direction by both $\vv{D}$ and $\vv{D}'$, we pick the embedding of $T_{w_D(e)}$ so that the boundary with positive contribution goes on the left; otherwise, if $e$ is oriented in the other direction by $\vv{D}'$, the boundary with positive contribution goes on the right. In this way, for every face $f$ of $\vv{G}$ that corresponds to a face in $\vv{D}$, the change in the total contribution (compared to $\vv{D}$) of the boundary path of $f$ is exactly $\sum_{e \in \outedg(v)} w(e) - \sum_{e' \in \inedg(v)} w(e')$ = 0, where $v$ in $\vv{P}$ is the dual vertex of the face $f$ in $\vv{D}$. Then Theorem~\ref{theorem:upwards} can be invoked to show that the constructed embedding is upward planar. We will show that nearly all of the conditions of the Theorem are satisfied automatically by this construction, including the conditions \textbf{UP0--UP2} on all vertices and conditions \textbf{UP3} on the inner faces of the tendrils; the only non-trivial conditions are \textbf{UP3} on the faces originating in the faces of $\vv{D}$, and their fulfillment is guaranteed precisely by the fact that the change in the total contribution along the boundary walk for all such faces is zero, as highlighted earlier.
\end{proof}

\section{From \textsc{Circulating Orientation} to \textsc{Rectilinear Planarity Testing}}
%\todo[inline]{This section is also ready for proofreading}
In this section we present the reduction from \textsc{Circulating Orientation} to \textsc{Rectilinear Planarity Testing}.
\begin{lemma} \label{lemma:orthogonal}
There is a polynomial-time reduction from planar triangulated triconnected instances of \textsc{Circulating Orientation} of pathwidth~$k$ with capacities polynomial in graph size to \textsc{Rectilinear Planarity Testing} instances of pathwidth~$\Oh(k)$.
\end{lemma}
\begin{proof}
    Let $(P, c)$ be an instance of \textsc{Circulating Orientation}, where $P$ is a triconnected planar triangulated graph, and $c: E(P) \to \mathbb{Z}_{\ge 0}$ is the capacity function.
    %By Theorem~\ref{theorem:pw_triangulation}, w.l.o.g. we can assume that $P$ is triangulated, as adding zero-weight edges to an instance of \textsc{Circulating Orientation} is safe.
    Let $\Epsilon_P$ be a planar embedding of $P$, and let $D$ be the dual graph of $P$ with respect to $\Epsilon_P$. Let $\Epsilon_D$ be the respective planar embedding of $D$, and let $w : E(D) \to \mathbb{Z}_{\ge 0}$ be the weight function on the edges of $D$ defined so that $w(e') = c(e)$, where $e \in E(P)$, $e' \in E(D)$, and $e'$ is the dual edge of $e$. Analogously to the proof of~\cref{lemma:upwards}, we have the following.
    
    \begin{numclaim}
        The graph $D$ is a planar triconnected graph of maximum degree 3 and pathwidth $\Oh(k)$.
    \end{numclaim}

    We now construct the graph $F$ from $D$ by subdividing every edge 4 times. Let $\rep: E(D) \to E(P)$ be a mapping that associates every edge in $D$ to the middle edge of the respective subdivision in $P$; we call the edge $\rep(e) \in E(P)$ a \emph{representative} of $e \in E(D)$.
    \begin{numclaim}
        The graph $F$ has a unique planar embedding, and it admits a rectilinear embedding.
    \end{numclaim}
    \begin{proof}
        Identical to Lemma 5.1 of~\cite{GargT01}.
    \end{proof}

    We set a large enough parameter $\theta = |V(F)| + 1$.
    We then construct a graph $G$ from the graph $F$ by replacing every representative edge $\rep(e)$ by a rectilinear tendril $T_k$, where $k = \theta \cdot w(e)$. We call the internal vertices of the tendrils \emph{auxiliary} in $G$, and the other vertices \emph{base vertices}. Base vertices are associated with vertices of $F$ in the natural way.
    Similarly to the proof of~\cref{lemma:upwards}, since $D$ is triconnected and the rectilinear tendrils admit a single planar embedding, the embeddings of $G$ are determined by the flips of the rectilinear tendrils. The faces of an embedding of $G$ are thus either internal in the tendrils, or originate from the faces of $D$. We call the former \emph{auxiliary} and the latter \emph{base faces} of the respective embedding of $G$.
    We also call the angles at auxiliary vertices \emph{auxiliary angles}, and the angles at base vertices \emph{base angles}.

    In the same way as in~\cref{lemma:upwards}, we establish a correspondence between planar embedding of $G$ and orientations of $D$, and between orientations of $D$ and orientations of $P$. From an orientation $\vv{P}$ of $P$ we construct an orientation $\vv{D}$ of $D$: for an arc $uv \in E(\vv{P})$, we orient the dual edge $e$ clockwise in the order of the boundary walk of $f_u$, which is the dual face of $u$ in $D$; here clockwise is given by the unique rectilinear embedding of $D$. Naturally, this transformation could be reversed, so we obtain a bijection between orientations of $P$ and of $D$; under this bijection, the sum of capacities of outgoing (incoming) arcs from a vertex $v$ in $P$ is equal to the sum of weights of the arcs oriented clockwise (counter-clockwise) along the boundary of the dual face $f_v$ in $D$.
    
    In the second part of the correspondence, we pick a flip of each tendril in $G$ in accordance with the orientation $\vv{D}$ of $D$: for a face $f$ in $D$, if an edge of the boundary is oriented clockwise, the tendril is flipped so the boundary walk with the positive contribution is towards $f$; so counter-clockwise edges result in the negative contribution of the respected tendril. In this way, for a base face $f$ of $G$, the total significant contribution of the tendrils to $f$ is equal to $4\theta$ times the difference between the total weight of clockwise edges and the total weight of counter-clockwise edges. This transformation is also two-way: a planar embedding of $G$ is defined by picking a flip for each tendril, and orienting the edges clockwise whenever the respective boundary walk has positive contribution to the respective face produces the matching orientation of $D$. In the following, we denote the embedding of $G$ produced from an orientation of $\vv{D}$ of $D$ in the manner above by $\Epsilon_{\vv{D}}$.

    \begin{numclaim}\label{claim:orth_significant}
        Let $f$ be a base face of a rectilinear embedding of $G$. Then the total significant contribution to $f$ of its tendrils is $0$.
    \end{numclaim}
    \begin{proof}
        There are no vertices of degree 1 in $G$, thus no angle receives the label $4$ in the rectilinear embedding, otherwise the property \textbf{RE0} of~\cref{th:rect-conditions} is violated. Thus, the total contribution to $f$ is $n_3(f) - n_1(f)$. By \textbf{RE1}, $n_3(f) - n_1(f)$ is either $4$ or $-4$.

        The total contribution $n_3(f) - n_1(f)$ consists of base angles and auxiliary angles, where the total contribution of the latter is exactly the total contribution of the tendrils. Let $|f|$ be the number of base vertices on the boundary of $f$, which is equal to the number of vertices on the boundary of $f$ in $F$ and to the number of tendrils on the boundary of $f$ in $G$. The total contribution of base angles to $f$ is then between $-|f|$ and $|f|$. The difference between the total contribution and total significant contribution of all tendrils to $f$ is between $-2|f|$ and $2|f|$, since for every tendril the difference lies between $-2$ and $2$. As the total contribution to $f$ must be $4$ in absolute value, the total significant contribution of the tendrils is at most $3|f| + 4$, since the other parts of the contribution are at most $3|f|$ in absolute value in total, as shown above. Assume now that the total significant contribution of tendrils to $f$ is non-zero, then its absolute value is at least $4\theta$, since the significant contribution of every tendril is a multiple of $4\theta$. However, $4\theta = 4(|V(F)| + 1) > 3|f| + 4$, which is a contradiction.
    \end{proof}
    
    It remains to show that $P$ has a circulating orientation if and only if $G$ admits a rectilinear embedding.

     Fix a rectilinear embedding of $G$, and consider a base face $f$. By~\cref{claim:orth_significant}, the total significant contribution of its tendrils is zero. Consider the respective orientation of $D$, in it $f$ has the same total weight of clockwise and counter-clockwise edges. In turn, for the dual vertex $v_f$ in $P$, orienting the dual in $P$ of every clockwise edge outwards from $v_f$ and the dual of every counter-clockwise edge towards $v_f$ gives that $\sum_{e \in E^+_{\vv{P}}(v_f)} c(e) = \sum_{e' \in E^-_{\vv{P}}(v_f)} c(e')$. Since this holds for every vertex $v$ of $P$, the constructed orientation is a solution to $(P, c)$.

     In the other direction, consider an orientation $\vv{P}$ of $P$ such that for every vertex $v \in V(P)$, $\sum_{e \in E^+_{\vv{P}}(v)} c(e) = \sum_{e' \in E^-_{\vv{P}}(v)} c(e')$. This gives an orientation $\vv{D}$ of $D$ where every face has the same total weight of clockwise and counter-clockwise edges. Consider the respective embedding $\Epsilon_{\vv{D}}$ of $G$, we claim that it is rectilinear by showing a suitable angle assignment. Assign angles to the auxiliary angles as per the rectilinear embedding of the tendril where the contribution is equal to the significant contribution, and base angles as per the rectilinear embedding of $F$. The condition \textbf{RE0} of~\cref{th:rect-conditions} is thus fulfilled automatically. Now, consider a face $f$ in $G$.
     For an auxiliary face, the condition \textbf{RE1} on the total contribution is again fulfilled automatically since the angle assignment on the boundary of the face matches the angle assignment in a rectilinear embedding of the tendril. For a base face $f$, its total contribution is equal to the total contribution in the rectilinear embedding of $F$, as the total contribution of the tendrils on its boundary is equal to their total significant contribution which is equal to zero. Thus the condition \textbf{RE1} is also fulfilled in this case, therefore by~\cref{th:rect-conditions} the constructed angle assignments gives a rectilinear embedding of $G$.    
\end{proof}

%\bmpr{Add some filltext to say how the arguments above yield the lemma below. Potentially split in two lemmata, one for upwards and one for rectilinear, and then say we get them both so the lemma from the main text follows.}

Finally, we obtain the main lemma, restated below for convenience, as a combination of reduction from \textsc{Multicolored Clique} to \textsc{Circulating Orientation} given by~\cref{lemma:co}, and the reductions from \textsc{Circulating Orientation} to the respective planarity testing problems given by \cref{lemma:upwards,lemma:orthogonal}.

\LemPlanarity*\label{lemma:planarity*}

% \bibliographystyle{plain}
% \bibliography{biblio}
\end{document}